\documentclass[10pt,twocolumn,twoside,submit]{JCNtran}
\usepackage[cmex10]{amsmath}
\usepackage{amssymb}

\usepackage{cases}
\usepackage{multirow}
\usepackage{empheq}

\usepackage[noadjust]{cite}
\usepackage{verbatim}
\usepackage{lipsum}

\usepackage{epsfig}

\usepackage[caption=false,font=footnotesize]{subfig}
\usepackage{color}

\usepackage{algorithm}
\usepackage{algpseudocode}

\newtheorem{theorem}{Theorem}
\newtheorem{lemma}[theorem]{Lemma}

\newtheorem{corollary}[theorem]{Corollary}
\newtheorem{definition}[theorem]{Definition}

\newtheorem{proposition}[theorem]{Proposition}
\newtheorem{remark}[theorem]{Remark}

\DeclareMathOperator{\rank}{rank}

\begin{document}

\title{Redundancy Allocation in Finite-Length Nested Codes for Nonvolatile Memories}

\author{Yongjune~Kim and~B.~V.~K.~Vijaya~Kumar
\thanks{The material in this paper was presented in part at the IEEE Conference on Communications (ICC), Budapest, Hungary, June 2013 and the IEEE International Symposium on Information Theory (ISIT), Istanbul, Turkey, July 2013.
Y. Kim is with the Coordinated Science Laboratory, University of Illinois at Urbana-Champaign, Urbana, IL, 61801, USA and B. V. K. Vijaya Kumar is with Carnegie Mellon University Africa in Kigali, Rwanda (e-mail: yongjune@illinois.edu, kumar@ece.cmu.edu).}
}
\maketitle

\begin{abstract}
	In this paper, we investigate the optimum way to allocate redundancy of finite-length nested codes for modern nonvolatile memories suffering from both \emph{permanent} defects and \emph{transient} errors (erasures or random errors). A nested coding approach such as partitioned codes can handle both permanent defects and transient errors by using two parts of redundancy: 1) redundancy to deal with permanent defects and 2) redundancy for transient errors. We consider two different channel models of the binary defect and erasure channel (BDEC) and the binary defect and symmetric channel (BDSC). The transient errors of the BDEC are erasures and the BDSC's transient errors are modeled by the binary symmetric channel, respectively. Asymptotically, the probability of recovery failure can converge to zero if the capacity region conditions of nested codes are satisfied. However, the probability of recovery failure of finite-length nested codes can be significantly variable for different redundancy allocations even though they all satisfy the capacity region conditions. Hence, we formulate the redundancy allocation problem of finite-length nested codes to minimize the recovery failure probability. We derive the upper bounds on the probability of recovery failure and use them to estimate the optimal redundancy allocation. Numerical results show that our estimated redundancy allocation matches well the optimal redundancy allocation.
\end{abstract}


\section{\uppercase{Introduction}}

	The channel model of memory with defective cells (i.e., stuck-at defects) was introduced by Kuznetsov and Tsybakov~\cite{Kuznetsov1974}. As shown in Fig.~\ref{fig:BDC}, this channel model has a channel state $S \in \{0, 1, \lambda\}$, which is called defect information. The state $S=0$ corresponds to a stuck-at 0 defect that always outputs a 0 independent of its input value, the state $S=1$ corresponds to a stuck-at 1 defect that always outputs a 1, and the state $S = \lambda$ corresponds to a normal cell that outputs the same value as its input. The probabilities of 0, 1, $\lambda$ states are $\beta/2$, $\beta/2$ (assuming a symmetric defect probability), and $1 - \beta$, respectively~\cite{Kuznetsov1974, ElGamal2011}. We call this channel model the binary defect channel (BDC).  
	
	The capacity of the BDC is $1 - \beta$ when both the encoder and the decoder know the defect information. If the decoder is aware of the defect locations in an array of memory cells, then the defects can be treated as erasures so that the capacity is $1 - \beta$~\cite{Heegard1983capacity, ElGamal2011}. On the other hand, Kuznetsov and Tsybakov investigated the model where the encoder knows the channel state information (i.e., locations and stuck-at values of defects) and the decoder does not have any information of defects~\cite{Kuznetsov1974}. It was shown that the capacity is also $1 - \beta$ even when only the encoder knows the defect information~\cite{Kuznetsov1974, Heegard1983capacity}. The capacity of the BDC is given by 	
	\begin{equation}\label{eq:BDC_capacity}
		C_{\mathrm{BDC}} = 1 - \beta. 
	\end{equation}

	A practical coding scheme for the BDC is \emph{additive encoding} which masks defects by adding a carefully selected binary vector~\cite{Kuznetsov1974, Tsybakov1975additive, Dumer1990}. The goal of masking defects is to make a codeword whose values at the locations of defects match the stuck-at values of corresponding defects. 
	
	\begin{figure}[!t]
		\centering
		\includegraphics[width=0.3\textwidth]{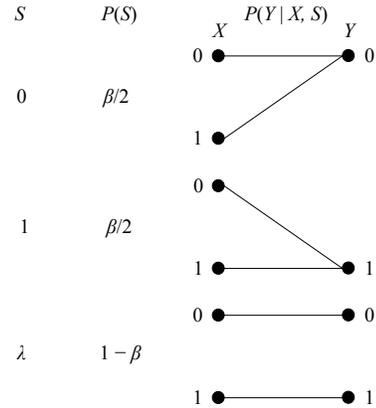}
		\vspace{-3mm}
		\caption{Binary defect channel (BDC).}
		\label{fig:BDC}
		\vspace{-5mm}
	\end{figure}


	Recently, the BDC has received renewed attention as a possible channel model for nonvolatile memories such as phase change memories (PCMs) and flash memories~\cite{Lastras-Montano2010,Jagmohan2010coding,Hwang2011iterative,Jacobvitz2013,Kim2013coding, Kim2013redundancy, Mahdavifar2015,Kim2016locally}. The write operation of PCM involves heating and cooling chalcogenide alloy, which results in expansion and contraction of the alloy. The repeated writes eventually cause the heating element to detach from the alloy resulting in a stuck-at defect that can be read but not written~\cite{Pirovano2004}. In flash memories, it has been shown that highly interfered cells can be regarded as stuck-at defects because of the unique properties of flash memory~\cite{Kim2014dirtyflash}. 
	
	In~\cite{Tsybakov1975}, the binary defect and symmetric channel (BDSC) model was considered where memory cells suffer from both stuck-at defects and random errors. This channel model is more realistic since random errors happen in PCM cells because of write noise, resistance drift, and unwanted heating~\cite{Pirovano2004}. Similarly, flash memory suffers from random errors due to charge loss and random telegraph noise~\cite{Prall2007}. 
	


	A \emph{nested coding} approach for the BDSC was proposed by Tsybakov~\cite{Tsybakov1975} and Heegard~\cite{Heegard1983plbc}, where the codes are referred to as \emph{partitioned linear block codes (PLBCs)}. Heegard showed that PLBC can achieve the capacity of the BDSC as well as suggested specific code constructions based on Bose-Chaudhuri-Hocquenghem (BCH) codes, i.e., partitioned BCH (PBCH) codes~\cite{Heegard1983plbc,Zamir2002nested}.
	
	A PLBC requires two generator matrices as other nested linear codes. One of them is for masking stuck-at defects and the other generator matrix is for correcting transient errors. Due to these two generator matrices, we can separate the redundancy for masking defects from the redundancy for correcting random errors~\cite{Heegard1983plbc}. We assume that the number of redundant bits for masking defects and correcting random errors are $l$ and $r$, respectively. Hence, the total redundancy is $l + r = n - k$ (where $n$ is the codeword size and $k$ is the message size). Note that the code rate is $R = k / n = (n - l - r) / n $.
	
	The fact that the redundancy can be divided into two parts leads to the problem of redundancy allocation. The objective is to find the optimal $(l, r)$ minimizing the probability of recovery failure for a fixed total redundancy $n-k$. This redundancy allocation problem can be stated as follows:
	\begin{equation}\label{eq:opt_prob}
		\begin{aligned}
		\left( l^*, r^* \right)  = \; & \underset{(l, r)}{\text{argmin}} & & P( \widehat{\mathbf{m}} \ne \mathbf{m})  \\
		& \text{subject to} & & 0 \le l \le n - k, \quad 0 \le r \le n - k\\
		&&&  l+r = n-k 
		\end{aligned}
	\end{equation}	
	where $\mathbf{m}$ and $\widehat{\mathbf{m}}$ denote a message and its estimate (recovered message), respectively. $P( \widehat{\mathbf{m}} \ne \mathbf{m})$ denotes the probability of recovery failure. Not surprisingly, the optimal redundancy allocation $(l^*, r^*)$ depends on the channel. If a channel exhibits only stuck-at defects, we should allot all redundancy to masking stuck-at defects, i.e., $(l^{*}, r^{*})=(n-k, 0)$. It is also expected that the optimal redundancy allocation for a channel with only random errors is $(l^{*}, r^{*})=(0, n-k)$.
		
 	Our main contributions in this paper are the formulation of redundancy allocation problem for finite-length PLBC and the techniques to estimate the optimal redundancy allocation. Asymptotically, the probability of recovery failure can converge to zero if the capacity conditions of PLBC are satisfied. However, we observe that the probability of recovery failure of finite-length PLBC can be significantly different depending on redundancy allocations even though they all satisfy the capacity region conditions. 
 	
 	We propose the following techniques to estimate the optimal redundancy allocation of finite-length PLBC. First, we investigate the redundancy allocation of the binary defect and erasure channel (BDEC) where transient errors are modeled by erasures. We derive the upper bound on the probability of recovery failure. Based on this upper bound, we will obtain the estimate $(\widehat{l}, \widehat{r})$ for the optimal $(l^{*}, r^{*})$. Similarly, we derive the estimate of the probability of recovery failure for the BDSC. Using this estimate of the probability of recovery failure, we can estimate the optimal redundancy allocation. Numerical results show that the estimates $(\widehat{l}, \widehat{r})$ of both the BDEC and the BDSC are well matched with the optimal redundancy allocations obtained by Monte-Carlo simulations.
	    
    To the best of our knowledge, the redundancy allocation problem for finite-length PLBC has not been addressed. Past work providing asymptotic results is not useful in trying to determine the optimal redundancy allocation for finite-length PLBC. We believe that this redundancy allocation problem is practically important with the applications of nonvolatile memories. 
	
	The rest of the paper is as follows. In Section~\ref{sec:channel}, we explain the channel models and PLBC. In Section \ref{sec:ra}, we propose techniques to determine the proper redundancy allocation of finite-length partitioned codes for the BDEC and the BDSC. Section~\ref{sec:conclusion} concludes the paper.

\section{\uppercase{Channel Model and PLBC: Preliminaries}} \label{sec:channel}

\subsection{Channel Model}

	The BDC in Fig.~\ref{fig:BDC} can be described as follows. Let ``$\circ$'' denote the operator  $\circ:\{0, 1\} \times \left\{0, 1, \lambda\right\} \rightarrow \{0, 1\}$ as in~\cite{Heegard1983plbc}
	\begin{equation}\label{eq:circ_operator}
		c \circ s =
		\begin{cases}
		c, & \text{if } s = \lambda ; \\
		s, & \text{if } s \ne \lambda.
		\end{cases}
	\end{equation}
	By using the operator $\circ$, an $n$-cell memory with stuck-at defects is modeled by 
	\begin{equation}
	\mathbf{y} = \mathbf{c} \circ \mathbf{s} \label{eq:BDC_vector}
	\end{equation}
	where ${\mathbf{c}}, {\mathbf{y}}\in \left\{0, 1\right\}^n$ are the codeword and the channel output vector, respectively. Also, the channel state vector ${\mathbf{s}} \in \left\{ 0, 1, \lambda \right\}^n$ represents the defect information in the $n$-cell memory.  Note that $\circ$ is the vector component-wise operator.
	
	As shown in~Fig.~\ref{fig:BDC}, the probabilities of stuck-at defects and normal cells are given by 
	\begin{equation}\label{eq:BDC_prob}
		P(S = s) = 
		\begin{cases}
		1 - \beta, & \text{if } s = \lambda ; \\
		\frac{\beta}{2}, & \text{if } s \ne \lambda.
		\end{cases}
	\end{equation}

	Now we consider the channel models with both permanent stuck-at defects and transient errors, i.e., the BDEC and the BDSC. For the BDEC, $\mathbf{c} \circ \mathbf{s}$ will be the channel input of the BEC with the erasure probability $\alpha$. Also, the BDSC can be modeled by
	\begin{equation}
		\mathbf{y} = \mathbf{c} \circ \mathbf{s} + \mathbf{z} \label{eq:BDSC_vector}
	\end{equation}
	where $\mathbf{c} \circ \mathbf{s}$ is the channel input of the binary symmetric channel (BSC) with the crossover probability $p$ and ${\mathbf{z}} \in \{0, 1\}^n$ is the random error vector. 
	
	Assuming that stuck-at defects do not suffer from transient errors, the capacity can be determined by the following theorem. 	
	\begin{theorem}\cite{Heegard1983capacity}\label{thm:capacity_separate}
		If the stuck-at defects do not suffer from transient errors, then the capacity is given by
		\begin{equation} \label{eq:capacity_separate}
			\widetilde{C}^{\text{max}} = \widetilde{C}^{\text{enc}} = P(S = \lambda)\widetilde{C}
		\end{equation}
	\end{theorem}
	where $\widetilde{C}$ is the capacity of the discrete memoryless channel (DMC) with $P(Y \mid X) = P(Y \mid X, S = \lambda)$. The superscript `max' in \eqref{eq:capacity_separate} represents the capacity when the defect information is fully known to both the encoder and decoder. Also, the superscript `enc' denotes the capacity when only the encoder knows the defect information. Since $\widetilde{C}^{\text{enc}}$ is the same as $\widetilde{C}^{\text{max}}$, we can omit the superscript if the stuck-at defects do not suffer from transient errors. 
	
	Let $\widetilde{C}_{\text{BDEC}}$ denote the capacity of the channel when only the normal cells can be erased. Similarly, let $\widetilde{C}_{\text{BDSC}}$ denote the capacity of the channel when only the normal cells suffer from random errors. By Theorem \ref{thm:capacity_separate}, it is clear that
	\begin{align}
		\widetilde{C}_{\text{BDEC}} & = (1-\beta)(1-\alpha), \label{eq:capacity_BDEC_tilde} \\
		\widetilde{C}_{\text{BDSC}} & = (1-\beta)(1-h(p)) \label{eq:capacity_BDSC_tilde}
	\end{align}
	where $h\left( x\right) = -x \log_2 x - \left(1-x\right) \log_2\left(1-x\right)$. If neither the encoder nor the decoder knows the defect information, the capacity is given by~\cite{Heegard1983capacity}
	\begin{equation} \label{eq:capacity_BDSC_tilde_min}
		\widetilde{C}_{\text{BDSC}}^{\text{min}} = 1-h\left((1-\beta)p + \frac{\beta}{2} \right),
	\end{equation}	
	which is the capacity of the BSC with the crossover probability 
	\begin{equation}\label{eq:p_tilde}
	\widetilde{p} = (1-\beta)p + \frac{\beta}{2}.
	\end{equation}			
	
	If both the stuck-at defects and the normal cells suffer from transient errors, it is difficult to derive closed-from expressions for $C_{\text{BDSC}}^{\text{max}}$ and $C_{\text{BDSC}}^{\text{enc}}$. Note that the superscript `max' represents the capacity when the defect information is fully known to both the encoder and decoder and the superscript `enc' represents the capacity when only the encoder knows the defect information. Instead, these capacities can be evaluated numerically~\cite{Heegard1983capacity}. For the BDEC, we can derive $C_{\text{BDEC}}^{\text{max}}$ and $C_{\text{BDEC}}^{\text{enc}}$ because of the channel model's simpleness. If both the encoder and decoder know the defect information, this channel is equivalent to the BEC with the erasure probability of $\alpha + \beta - \alpha \beta$. Hence, 
	\begin{equation}
	    C_{\text{BDEC}}^{\text{max}} = 1 - \alpha - \beta + \alpha \beta = (1 - \beta)(1 - \alpha)
	\end{equation}
	which is the same as in \eqref{eq:capacity_BDEC_tilde}. The capacity $C_{\text{BDEC}}^{\text{enc}}$ is derived as follows. 
	\begin{proposition} \label{thm:capacity_BDEC}
		If only the encoder knows the defect information, the capacity of the BDEC is given by
		\begin{equation} \label{eq:capacity_BDEC}
			C_{\text{BDEC}}^{\text{enc}} = 1 - \alpha - \beta
		\end{equation}
		which shows that $C_{\text{BDEC}}^{\text{enc}} < C_{\text{BDEC}}^{\text{max}} = \widetilde{C}_{\text{BDEC}}$.
	\end{proposition}
	\begin{proof}
		The proof is given in Appendix \ref{pf:capacity_BDEC}-A. 
	\end{proof}	
	Although $C_{\text{BDEC}}^{\text{enc}} < C_{\text{BDEC}}^{\text{max}} = \widetilde{C}_{\text{BDEC}}$, the difference between $C_{\text{BDEC}}^{\text{enc}}$ and $\widetilde{C}_{\text{BDEC}}$ is only $\alpha \beta$ which is much smaller than $\alpha$ or $\beta$ for $\alpha, \beta \ll 1$. 
	
	For the BDSC, the closed-form expressions for $C_\text{BDSC}^{\text{max}}$ and $C_\text{BDSC}^{\text{enc}}$ are not known. In~\cite{Heegard1983plbc}, only a non-closed-form expression for $C_\text{BDSC}^{\text{enc}}$ was derived. Although we cannot obtain a closed-form expression for $C_\text{BDSC}$, we derive the following bound. 	
	\begin{proposition} \label{thm:capacity_bound_BDSC}
		If only the encoder knows the defect information, the capacity of the BDSC is bounded by
		\begin{equation} \label{eq:capacity_bound_BDSC}
			C_{\text{BDSC}}^{\text{lower}} \le C_{\text{BDSC}}^{\text{enc}} \le C_{\text{BDSC}}^{\text{upper}}			 
		\end{equation}
		where 
		\begin{align} 
			C_{\text{BDEC}}^{\text{lower}} &= 1 - \beta - h(p), \label{eq:capacity_lower_BDSC} \\
            C_{\text{BDEC}}^{\text{upper}} &= \widetilde{C}_{\text{BDSC}} = (1 - \beta)(1 - h(p)). \label{eq:capacity_upper_BDSC} 
		\end{align}		
	\end{proposition}
	\begin{proof}
		The proof is given in Appendix \ref{pf:capacity_bound_BDSC}-B.
	\end{proof}	
	From \eqref{eq:capacity_lower_BDSC} and \eqref{eq:capacity_upper_BDSC}, we can see that the difference between the upper bound and lower bound is only $\beta h(p)$. For $\beta, p \ll 1$, $\beta h(p)$ is much smaller than $\beta$ and $h(p)$. 	
	

\subsection{Partitioned Linear Block Codes (PLBCs)} \label{subsec:PLBC}

	In \cite{Heegard1983plbc}, Heegard proposed the PLBCs which can be viewed as nested codes~\cite{Zamir2002nested}. PLBCs are capable of correcting both stuck-at errors (due to stuck-at defects) and transient errors. Note that the PLBC can correct transient errors occurring in both stuck-at defects and normal cells. An $[n,k,l]$ PLBC is a pair of linear subspaces ${\mathcal{C}}_1 \subset \{0, 1\}^n$ and ${\mathcal{C}}_0 \subset \{0, 1\}^n$ of dimension $k$ and $l$ such that ${\mathcal{C}}_1 \cap {\mathcal{C}}_0 =\{ \mathbf{0}\}$. Then the direct sum
	\begin{equation}\label{eq:direct_sum}
	{\mathcal{C}} \triangleq {\mathcal{C}}_1 + {\mathcal{C}}_0 = \{ {\mathbf{c}} = {\mathbf{c}}_1 + {\mathbf{c}}_0 | {\mathbf{c}}_1 \in {\mathcal{C}}_1 , {\mathbf{c}}_0 \in {\mathcal{C}}_0 \}. 
	\end{equation}
	is an $[n, k+l]$ linear block code with a generator matrix $\widetilde{G} = \left[ G_1 \quad G_0 \right]$. Note that the sizes of $G_1$ and $G_0$ are $n \times k$ and $n \times l$, respectively. An example of PLBC can be found in~\cite{Heegard1983plbc}.  
	
	The encoding and decoding are described as follows.		
	\emph{Encoding:} A message $\mathbf{m}$ is encoded to a corresponding codeword $\mathbf{c} = \mathbf{c}_1 + \mathbf{c}_0 = G_1\mathbf{m} + G_0\mathbf{d}$. First, we encode $\mathbf{m}$ into $G_1 \mathbf{m}$ for correcting transient errors. Next, the stuck-at defects are masked by $G_0 \mathbf{d}$ where $\mathbf{d}$ is chosen carefully to mask stuck-at defects. 
	
	\emph{Decoding:} Let the received vector be $\mathbf{y} = \mathbf{c} \circ \mathbf{s} + \mathbf{z}$ according to~\eqref{eq:BDSC_vector}. The decoder computes the syndrome ${\mathbf{v}} = \widetilde{H}^T {\mathbf{y}}$ (superscript $T$ denotes transpose) where $\widetilde{H}$ is the parity check matrix of $\mathcal{C}$ such that $\widetilde{H}^T \widetilde{G} = 0_{r, k+l}$ (the $r \times (k+l)$ all-zero matrix). Based on the syndrome $\mathbf{v}$, the decoder chooses $\widehat{\mathbf{z}}$ (i.e., the estimate of $\mathbf{z}$) such that $\widetilde{H}^T \widehat{\mathbf{z}} = \mathbf{v}$. Then $\widehat{\mathbf{m}} =  \widetilde{G}_1^T \widehat{\mathbf{c}}$ where $\widehat{\mathbf{c}}= \mathbf{y} + \widehat{\mathbf{z}}$. The message inverse matrix $\widetilde{G}̃_1$ is defined as an $n \times k$ matrix such that $\widetilde{G}̃_1^{T} G_1 =I_k$ and $\widetilde{G}̃_1^{T} G_0 =0_{k, l}$.

	
	\begin{definition}
		A pair of minimum distances $\left(d_0, d_1 \right)$ of an $[n, k, l]$ PLBC are given by \cite{Heegard1983plbc}. 
		\begin{align} \label{eq:PLBC_d0}
		d_0 = \underset{
			\substack{
				{\mathbf{c}} \ne {\mathbf{0}} \\
				G_0^T {\mathbf{c}}= {\mathbf{0}}
			}}
			{\text{min }} \|{\mathbf{c}}\|, \:\:
			d_1 = \underset{
				\substack{
					\mathbf{m}  \ne \mathbf{0} \\
					\widetilde{H}^T \mathbf{c}= \mathbf{0}}}
			{\text{min }} \|\mathbf{c}\|
			\end{align}
			where $\| \cdot \|$ is the Hamming weight of the given vector. 
	\end{definition}
	Note that $d_1$ is greater than or equal to the minimum distance of the $[n,k+l]$ linear block code with the parity check matrix $\widetilde{H}$, while $d_0$ is the minimum distance of the $[n,k+r]$ linear block code with the parity check matrix $G_0$~\cite{Heegard1983plbc}. 
	
%
%
%

	It is critical to choose the proper $\mathbf{d}$ during the encoding stage of PLBC. The minimum distance encoding (MDE) chooses $\mathbf{d}$ as follows. 
	\begin{align}
		\mathbf{d}^* & =\underset{\mathbf{d}}{\text{argmin }} \|\mathbf{c} \circ \mathbf{s} - \mathbf{c} \| =  \underset{\mathbf{d}}{\text{argmin }} \left\| \mathbf{s}^{\mathcal{U}} + G_0^{\mathcal{U}} \mathbf{d} +  G_1^{\mathcal{U}} \mathbf{m}   \right\| \nonumber \\
		&= \underset{\mathbf{d}}{\text{argmin }} \left\| G_0^{\mathcal{U}} \mathbf{d} + \mathbf{b}^{\mathcal{U}} \right\| \label{eq:BDC_opt_problem}
	\end{align}
	where $\| {\mathbf{c}} \circ {\mathbf{s}} - {\mathbf{c}} \|$ represents the number of errors due to stuck-at defects and ${\mathbf{b}} = G_1 {\mathbf{m}} + {\mathbf{s}}$. Also, ${\mathcal{U}}=\left\{i_1,\ldots,i_u \right\}$ indicates the set of locations of stuck-at defects. We use the notation of ${\mathbf{s}}^{\mathcal{U}}=\left(s_{i_1},\ldots,s_{i_u}\right)^T$, $G_0^{\mathcal{U}}=\left[{\mathbf{g}}_{0,i_1}^T,\ldots,{\mathbf{g}}_{0,i_u}^T \right]^T$, and $G_1^{\mathcal{U}}=\left[{\mathbf{g}}_{1,i_1}^T,\ldots,{\mathbf{g}}_{1,i_u}^T \right]^T$ where ${\mathbf{g}}_{0,i}$ and ${\mathbf{g}}_{1,i}$ are the $i$-th rows of $G_0$ and $G_1$ respectively. 
	
	By solving the optimization problem in \eqref{eq:BDC_opt_problem}, we can minimize the number of errors due to stuck-at defects. It was shown~\cite{Heegard1983plbc} that the capacity of the BDSC can be achieved by the MDE and the maximum likelihood decoding (MLD). However, the encoding complexity of the MDE is ${\mathcal{O}}(2^u)$. Instead of solving the exponential complexity optimization problem, we just try to solve the following linear equation.
	\begin{equation}\label{eq:BDC_LE}
	G_0^{\mathcal{U}} {\mathbf{d}} = {\mathbf{b}}^{\mathcal{U}}
	\end{equation}
	Gaussian elimination or some other linear equation solution methods can be used to solve \eqref{eq:BDC_LE} with ${\mathcal{O}}\left(u^3\right)$ complexity. If the encoder fails to find a solution of \eqref{eq:BDC_LE}, then encoding failure is declared. It is clear that \eqref{eq:BDC_LE} has at least one solution if and only if
	\begin{equation} \label{eq:BDC_LE_sol_exist}
	\rank \left( G_0^{\mathcal{U}} \right) = \rank \left( G_0^{\mathcal{U}} \mid \mathbf{b}^{\mathcal{U}} \right)
	\end{equation}
	where $\left( G_0^{\mathcal{U}} \mid \mathbf{b}^{\mathcal{U}} \right)$ denotes the augmented matrix. If $u<d_{0}$, $\rank\left(G_0^{\mathcal{U}}\right)$ is always $u$ by \eqref{eq:PLBC_d0}. Hence, \eqref{eq:BDC_LE_sol_exist} holds and at least one solution $\mathbf{d}$ exists. If $u \ge d_0$, the encoder may fail to find a solution of \eqref{eq:BDC_LE}. 
	
	For convenience, we define a random variable $E$ as follows.
	\begin{equation} \label{eq:BDC_E}
	E =
	\begin{cases}
	1, & \|\mathbf{c} \circ \mathbf{s} - \mathbf{c} \| = 0 \text{ (encoding success)} \\
	0, & \|\mathbf{c} \circ \mathbf{s} - \mathbf{c} \| \ne 0 \text{ (encoding failure)}
	\end{cases}
	\end{equation}	
	
	We observe that the probability of encoding failure $P(E=0)$ by solving \eqref{eq:BDC_opt_problem} is the same as $P(E=0)$ by solving \eqref{eq:BDC_LE}. It is because $G_0^{\mathcal{U}} \mathbf{d} \ne  \mathbf{b}^{\mathcal{U}}$ if and only if $\|{\mathbf{c}} \circ {\mathbf{s}} - {\mathbf{c}} \| \ne 0$. Although solving \eqref{eq:BDC_LE} is suboptimal in regards to the number of stuck-at errors (i.e., $\|{\mathbf{c}} \circ {\mathbf{s}} - {\mathbf{c}} \|$), $C_{\text{BDC}}$ can be achieved by solving \eqref{eq:BDC_LE} instead of \eqref{eq:BDC_opt_problem}, which is shown in the following Proposition by using the results of~\cite{Heegard1983plbc}. 
	
	\begin{proposition} \label{thm:BDC_capacity}
		$C_{\text{BDC}}$ can be achieved by solving \eqref{eq:BDC_LE}.
	\end{proposition}
	\begin{proof}The proof is given in Appendix~\ref{pf:BDC_capacity}-C. 
	\end{proof}
	
	We note that Dumer proposed an asymptotically optimal code for the BDC with encoding complexity ${\mathcal{O}}\left(n \log_2^3{n}\right)$~\cite{Dumer1990}. In this paper, we are focusing on Heegard's partitioned codes with encoding by solving linear equation of \eqref{eq:BDC_LE}. It is because this linear equation approach allows to derive the upper bound on encoding failure probability based on linear algebra. This upper bound plays a pivotal role in estimating the optimal redundancy allocation for the BDEC and the BDSC. Moreover, for a finite-length $n$ and a reasonable defect probability $\beta$, $u \simeq \beta n$ is not a large number, hence, the computational complexity ${\mathcal{O}}(u^3)$ is not high.   
	


\section{\uppercase{Redundancy Allocation of Finite-Length PLBC}} \label{sec:ra}

	In this section, we investigate the redundancy allocation for finite-length PLBC. In order to clarify the redundancy allocation problem for finite-length PLBC, we define a pair of code rates $\left(R_1, R_0\right)$ where $R_0 = \frac{l}{n}$ is the code rate of ${\mathcal{C}}_0$ in \eqref{eq:direct_sum}. Also, $R_1 = \frac{k}{n}$ is the code rate of ${\mathcal{C}}_1$, which is equivalent to the actual code rate $R = \frac{k}{n}$. For the given codeword size $n$, we can readily calculate $(l, r)$ from $\left(R_1, R_0\right)$. 
	
	\begin{figure}[!t]
		\vspace{-2mm}
		\centering
		\includegraphics[width=0.30\textwidth]{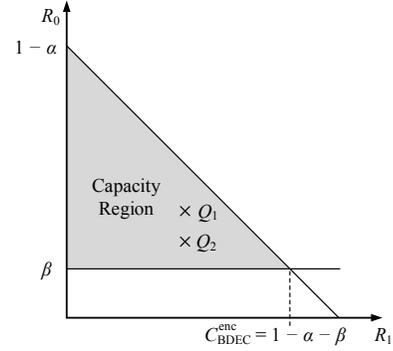}
		\vspace{-3mm}
		\caption{Capacity region of the BDEC derived in Proposition~\ref{thm:BDEC_region}. Two points of $Q_1$ and $Q_2$ in the capacity region represent the pairs of code rates $\left(R = R_1 = \frac{k}{n}, R_0 = \frac{l}{n}\right)$. Note that they have different redundant allocations $(l, r)$ in spite of the same actual code rates $R = R_1 = \frac{k}{n}$.}
		\label{fig:BDEC_region}
		\vspace{-5mm}
	\end{figure}

	For the BDEC, we will derive the capacity region $\mathbb{C}_{\text{BDEC}}$ in Theorem~\ref{thm:BDEC_region}. If $(R_1, R_0) \in \mathbb{C}_{\text{BDEC}}$, then $P(\widehat{\mathbf{m}} \ne \mathbf{m})$ converges to zero as $n \rightarrow \infty$. In Fig.~\ref{fig:BDEC_region}, $Q_1$ and $Q_2$ represent two pairs of code rates in ${\mathbb{C}}_{\text{BDEC}}$ with the same $R_1 = R$. Then, $Q_1$ and $Q_2$ have the same total redundancy $n - k = l + r$ whereas they have different redundancy allocations $(l, r)$. Note that $Q_1$ allocates more redundancy for defects (i.e., larger $l$) than $Q_2$. 
	
	Asymptotically, both $Q_1$ and $Q_2$ make $P(\widehat{\mathbf{m}} \ne \mathbf{m})$ converge to zero. On the other hand, $P(\widehat{\mathbf{m}} \ne \mathbf{m})$ of $Q_1$ and $Q_2$ can be significantly different for the finite-length codes, which leads to the need for formulating and solving the redundancy allocation problem in \eqref{eq:opt_prob}. 
	
	In order to choose the optimal redundancy allocation $(l^*, r^*)$ of \eqref{eq:opt_prob}, we should derive the $P(\widehat{\mathbf{m}} \ne \mathbf{m})$ which depends on the channel parameters (i.e., $\beta$, $\alpha$, $p$) as well as the code parameters $(n, k, l)$.  In the following subsection, we will derive the upper bound on encoding failure probability for finite-length codes, which is important for deriving the estimate of $P(\widehat{\mathbf{m}} \ne \mathbf{m})$. 
	
\subsection{Upper Bound on Encoding Failure Probability}
	
	Since we focus on the redundancy allocation problem in finite-length codes, we derive an upper bound on $P(E=0)$ for finite $n$. During the encoding stage, $\mathbf{d}$ is chosen by solving the linear equation \eqref{eq:BDC_LE} instead of the optimization problem \eqref{eq:BDC_opt_problem}. 
	
	\begin{lemma}\label{lemma:BDC_ef_UB_fixed_u} An upper bound on the probability of encoding failure given $u$ defects is given by
	\begin{equation}\label{eq:enc_fail_upper}
	P\left(E=0 \mid U=u \right) \le \min \left\{ \frac{\sum_{w=d_{0}}^{u}{B_{0, w} \binom{n-w}{u-w}}}{\binom{n}{u}}, 1\right\}
	\end{equation}
	where the random variable $U$ represents the number of stuck-at defects. In addition, $B_{0, w}$ is the weight distribution of ${\mathcal{C}}_{0}^{\perp}$ (i.e., the dual code of ${\mathcal{C}}_0$).
	\end{lemma}
	\begin{proof}
		The proof is given in Appendix~\ref{pf:BDC_ef_UB_fixed_u}-D.
	\end{proof}	
	
	Lemma~\ref{lemma:BDC_ef_UB_fixed_u} supports that $P\left(E=0 \mid U=u \right) = 0$ for $u < d_{0}$. The following Lemma states that $P\left(E=0 \mid U=u \right)$ can be obtained exactly for $d_0 \le u \le d_0 + \left\lfloor \frac{d_0-1}{2} \right\rfloor$ where $\lfloor x \rfloor$ is the largest integer not greater than $x$.
	
	\begin{lemma}\label{lemma:BDC_ef_exact}For $u \le d_0 + \left\lfloor \frac{d_0 - 1}{2} \right\rfloor$, $P\left(E=0 \mid U=u \right)$ is given by
	\begin{equation}
	\label{eq:BDC_ef_exact}
	P\left(E=0 \mid U=u \right) = \frac{1}{2} \cdot \frac{\sum_{w=d_{0}}^{u}{B_{0, w} \binom{n-w}{u-w}}}{\binom{n}{u}}.
	\end{equation}
	\end{lemma}
	\begin{proof}
	The proof is given in Appendix~\ref{pf:BDC_ef_exact}-E.
	\end{proof}

	From the definition of $d_0$ in \eqref{eq:PLBC_d0}, Lemma~\ref{lemma:BDC_ef_UB_fixed_u}, and Lemma~\ref{lemma:BDC_ef_exact}, we can state the following theorem.

	\begin{theorem}\label{thm:enc_fail_bound_fixed_u} $P\left(E = 0 \mid U=u\right)$ is given by		
		\begin{numcases}{\hspace{-6mm}}
		\hspace{-1mm}0, &\hspace{-6mm} for $u < d_0$ \label{eq:BDC_ef_exact_0}
		\\
		\hspace{-1mm}\frac{1}{2} \cdot \frac{\sum_{w=d_{0}}^{u}{B_{0, w} \binom{n-w}{u-w}}}{\binom{n}{u}},  &\hspace{-6mm} for $d_0 \le u \le d_0 + t_0$ \label{eq:BDC_ef_exact_1}
		\\
		\hspace{-1mm}\le \min \left\{ \frac{\sum_{w=d_{0}}^{u}{B_{0, w} \binom{n-w}{u-w}}}{\binom{n}{u}}, 1\right\}\hspace{-1mm}, &\hspace{-6mm} for $ u > d_0 + t_0$. \label{eq:BDC_ef_bound}
		\end{numcases}
		where $t_0 = \left\lfloor \frac{d_0 - 1}{2} \right\rfloor$.
	\end{theorem}
	
	We compare our upper bounds and simulation results assuming the the number of defects $u$ is given. The $[n = 31, k, l]$ partitioned BCH (PBCH) codes are considered and the weight distributions $B_{0, w}$ are calculated using the MacWilliams identity~\cite{Macwilliams1977theory}. The details of PBCH codes can be found in~\cite{Heegard1983plbc}. Fig.~\ref{fig:plot_fixed_u} shows that the upper bounds of \eqref{eq:BDC_ef_bound} are close to the simulation results for $P(E=0 \mid U=u)$. In addition, the calculated values of \eqref{eq:BDC_ef_exact_1} are well matched with the simulation results. Fig.~\ref{fig:plot_fixed_u} shows that the upper bounds approach $P(E=0 \mid U=u)$ and meet $P(E=0 \mid U=u)$ as the code rate decreases.
	
	From the bound on $P(E=0 \mid U=u)$ in Theorem~\ref{thm:enc_fail_bound_fixed_u}, we derive the following upper bound.
	
	\begin{corollary}\label{cor:BDC_UB_1} An upper bound on the probability of encoding failure $P\left(E = 0 \right)$ is given by 
	\begin{align}
	P(E=0)
	&\le \sum_{u=d_0}^{n}{ \beta^{u} \left( 1 - \beta \right)^{n-u}}\sum_{w=d_0}^{u}{B_{0,w} \binom{n-w}{u-w}} \label{eq:bound_enc_fail_2}.	
	\end{align}
	\end{corollary}
	\begin{proof}	
		It can be shown by $P(U=u) = \binom{n}{n} \beta^u (1-\beta)^{n-u}$ and \eqref{eq:BDC_ef_bound}. 
	\end{proof}

	\begin{figure}[!t]
	\centering
	\subfloat[$u=10$]{\includegraphics[width=0.4\textwidth]{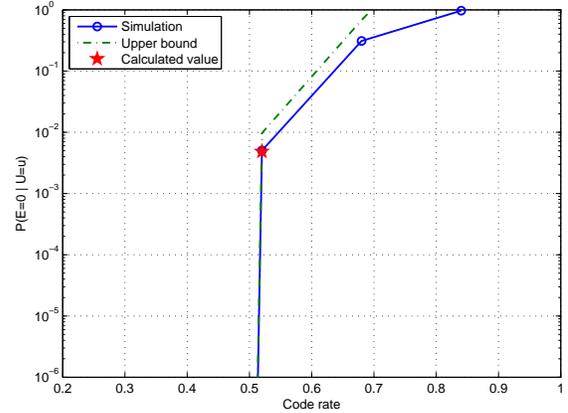}
		\label{fig:plot_fixed_u_10}}
	\hfill 
	\subfloat[$u=12$]{\includegraphics[width=0.4\textwidth]{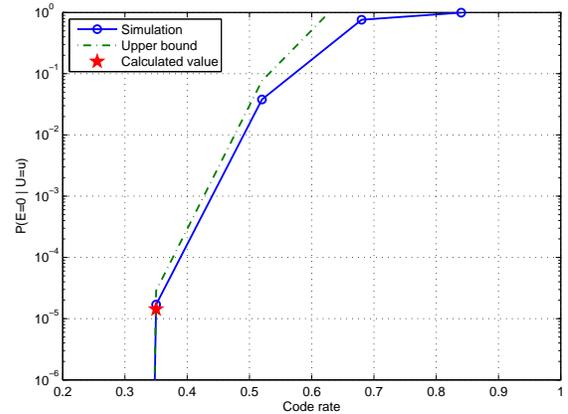} 
		\label{fig:plot_fixed_u_12}}
	\caption{Comparison of simulation results, upper bounds by \eqref{eq:BDC_ef_bound}, and calculated values by \eqref{eq:BDC_ef_exact_1} for $P(E = 0 \mid U=u)$. $[n = 31, k, l]$ PBCH codes are used. The code rate is $R = k/n$.}
	\label{fig:plot_fixed_u}
	\vspace{-5mm}
	\end{figure}
	
	Since it is intractable to compute $B_{0,w}$ for large $n$, we consider the following binomial approximation.
	\begin{equation}\label{eq:B0w_approximate}
		B_{0,w} \cong 2^{-l} \binom{n}{w}
	\end{equation}
	For many codes including random linear codes (each element of generator matrix is chosen uniformly at random from $\left\{0, 1 \right\}$) and BCH codes, it is known that the weight distributions are well approximated by the binomial distribution~\cite{Macwilliams1977theory}. 
	
	\begin{corollary}\label{cor:BDC_UB_2}If the weight distribution $B_{0,w}$ follows the binomial approximation, the upper bound of \eqref{eq:bound_enc_fail_2} is given by
	\begin{align}
		P(E=0) &\le 2^{-l}\left(1 + \beta\right)^n \label{eq:BDC_UB_2}, \\
		\log_2{P(E=0)} &\le n \left\{R - \left(1 - \log_2(1 + \beta) \right) \right\}. \label{eq:BDC_UB_2_log}
	\end{align}
	\end{corollary}
	\begin{proof}
	From \eqref{eq:bound_enc_fail_2} and \eqref{eq:B0w_approximate}, the upper bound on $P\left(E = 0 \right)$ can be derived as follows.
	\begin{align}
	P\left(E=0\right)
	& \le \sum_{u = d_0}^{n}{\beta^u \left(1 - \beta \right)^{n-u}  \sum_{w=d_0}^{u}{B_{0, w} \binom{n-w}{u-w}} } \label{eq:bound_enc_fail_variant_0} \\
	& = 2^{-l} \sum_{u = d_0}^{n}{\beta^u \left(1 - \beta \right)^{n-u} \sum_{w=d_0}^{u}{\binom{u}{w} \binom{n}{u}}} \label{eq:bound_enc_fail_variant_2} \\
	& \le 2^{-l} \sum_{u = 0}^{n}{\binom{n}{u} \left(2\beta\right)^u \left(1 - \beta \right)^{n-u}} \label{eq:bound_enc_fail_variant_4} \\
	& = 2^{-l} \left(1 + \beta \right)^n \label{eq:RA_M}
	\end{align}
	where \eqref{eq:bound_enc_fail_variant_2} follows from $\binom{n}{w} \binom{n-w}{u-w} = \binom{u}{w} \binom{n}{u}$. Also, \eqref{eq:bound_enc_fail_variant_4} follows from the binomial theorem $\sum_{w = 0}^{u}{\binom{u}{w}} = 2^u$.
	\eqref{eq:BDC_UB_2_log} can be obtained by taking the logarithm.
	\end{proof}	
	
	Fig.~\ref{fig:plot_fixed_beta} shows that the upper bound is very close to the simulation results when the probability of encoding failure is low. In regards to the simulation results, we used PBCH codes for the BDC with $\beta = 0.1$.
	
	\begin{remark}
	\eqref{eq:BDC_UB_2} shows that the probability of encoding failure decreases as $l$ increases, whereas the probability of encoding failure increases as $\beta$ increases. For infinite-length codes, this upper bound is not tight since $1 - \log_2(1 + \beta)$ of \eqref{eq:BDC_UB_2_log} is less than $C_{\text{BDC}} = 1 - \beta$. However, this upper bound is tight for finite-length codes as shown in Fig.~\ref{fig:plot_fixed_beta}. Moreover, the upper bound in \eqref{eq:BDC_UB_2_log} is the linear function of $R$ where $n$ is the slope and $1 - \log_2(1 + \beta)$ is the $R$-intercept, which explicitly shows that longer codes improve the probability of encoding failure.   
	\end{remark}
	
	\begin{figure}[!t]
		\centering
		\includegraphics[width=0.4\textwidth]{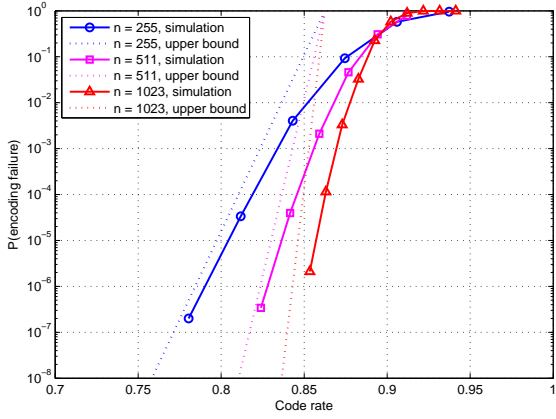}
		\caption{Comparison of simulation results and upper bounds in~\eqref{eq:BDC_UB_2} for the probability of encoding failure $P(E=0)$. We used PBCH codes for the BDC with probability of defect $\beta=0.1$.}
		\label{fig:plot_fixed_beta}
		\vspace{-2mm}
	\end{figure}
		

	\subsection{Redundancy Allocation: BDEC} \label{sec:BDEC_ra}

	For the BDEC, the encoder should try to mask stuck-at defects and the decoder should correct erasures. The encoding process of the BDEC is the same as the encoding in \ref{subsec:PLBC}. Only the decoding will be modified as follows. 
	
	\emph{Decoding (BDEC):} The MLD is done by solving the following linear equation.
	\begin{equation}\label{eq:BDEC_decoder}
		\widetilde{G}^{\mathcal{V}} \begin{bmatrix} \widehat{\mathbf{m}} \\ \widehat{\mathbf{d}} \end{bmatrix} = {\mathbf{y}}^{\mathcal{V}}
	\end{equation}
	where ${\mathcal{V}}=\left\{j_1,\cdots, j_v\right\}$ indicates the locations of $v$ unerased bits. We use the notation of ${\mathbf{y}}^{\mathcal{V}}=\left(y_{j_1}, \cdots, y_{j_v}\right)^T$ and $\widetilde{G}^{\mathcal{V}}=\left[ \widetilde{\mathbf{g}}_{j_1}^T, \cdots, \widetilde{\mathbf{g}}_{j_v}^T \right]^T$ where $\widetilde{\mathbf{g}}_j$ is the $j$-th row of $\widetilde{G}$. By solving \eqref{eq:BDEC_decoder}, we can obtain estimates of the message $\mathbf{m}$ and redundancy $\mathbf{d}$, i.e., $\widehat{\mathbf{m}}$ and $\widehat{\mathbf{d}}$. We consider the MLD in order to derive the upper bound on $P\left( \widehat{\mathbf{m}} \ne \mathbf{m} \right)$. The MLD can be accomplished by solving the linear equation in \eqref{eq:BDEC_decoder}, whose complexity is ${\mathcal{O}}\left(n^3\right)$. 
	
	
	
	\begin{theorem} \label{thm:BDEC_region}
		If a pair of code rates $\left(R_0, R_1\right)$ satisfy the following conditions, then the probability of recovery failure $P\left( \widehat{\mathbf{m}} \ne \mathbf{m} \right)$ approaches zero as $n \rightarrow \infty$.
		\begin{equation}
			R_0  > \beta, \quad R_1 + R_0 < 1 - \alpha \label{eq:BDEC_region}
		\end{equation}
	\end{theorem}
	where $\alpha$ is the probability of erasure and $\beta$ is the probability of defect for the BDEC. 
	
	\begin{proof}
		We show that $P\left( \widehat{\mathbf{m}} \ne \mathbf{m} \right)$ approaches zero with $n$ if $\left(R_0, R_1\right)$ satisfies \eqref{eq:BDEC_region}. We can claim that $P\left( \widehat{\mathbf{m}} \ne \mathbf{m} \right) = P\left(E=0\right) + P\left(E=1, D=0\right)$ where the random variable $D$ is defined as follows. 
		\begin{equation} \label{eq:BDEC_D}
		D =
		\begin{cases}
		1, & ~\text{decoding success} \\
		0, & ~\text{decoding failure}
		\end{cases}
		\end{equation}
		
		From~\eqref{eq:BDC_CA}, we can claim that $P\left( E=0 \right) \le n (\beta + \epsilon) 2^{-n \left( \frac{l}{n} - \beta - \epsilon\right)} + \epsilon'$. Hence, $P\left( E=0 \right)$ converges to zero if $R_0 = \frac{l}{n} > \beta$. If the additive encoding succeeds (i.e., $E=1$), then the corresponding channel is equivalent to the BEC with the erasure probability $\alpha$. It is because all the stuck-at defects are masked. Note that the code rate of $\widetilde{G}$ of \eqref{eq:BDEC_decoder} is $\frac{k+l}{n} = R_0 + R_1$. Thus, $R_0 + R_1 < 1 - \alpha$.
	\end{proof}
	
	
	\begin{table}[t]
		\caption{All Possible Redundancy Allocation Candidates of $\left[ n = 1023, k=923, l \right]$ PBCH Codes}
		\label{tab:PLBC}
		\centering
		\renewcommand{\arraystretch}{1.1}
		\small{
		\begin{tabular}{c|c|c|c|c|c}
			\hline
			Code & {$l$} & {$r$} & {$d_0$} & {$d_1$} & Remarks   \\ \hline \hline
			0 & 0 & 100 & 0 & 21 & Only correcting transient errors\\ \hline
			1 & 10 & 90 & 3 & 19 &\\ \hline
			2 & 20 & 80 & 5 & 17 &\\ \hline
			3 & 30 & 70 & 7 & 15 &\\ \hline
			4 & 40 & 60 & 9 & 13 &\\ \hline
			5 & 50 & 50 & 11 & 11 & \\ \hline
			6 & 60 & 40 & 13 & 9 &\\ \hline
			7 & 70 & 30 & 15 & 7 &\\ \hline
			8 & 80 & 20 & 17 & 5 &\\ \hline
			9 & 90 & 10 & 19 & 3 &\\ \hline
			10& 100 & 0 & 21 & 0 & Only masking stuck-at defects\\ \hline
		\end{tabular}}
		\vspace{-5mm}
	\end{table}	
	
	Fig.~\ref{fig:BDEC_region} represents the capacity region by Theorem~\ref{thm:BDEC_region}. The supremum of $R_1$ in this capacity region is $1 - \alpha - \beta$ which is equal to $C_{\text{BDEC}}^{\text{enc}}$ of \eqref{eq:capacity_BDEC}. From \eqref{eq:BDEC_region}, we can obtain $l  > n\beta$ and $r > n \alpha$, which achieve the capacity for infinite $n$.

	As explained earlier, these asymptotic results cannot be directly used to choose the optimal redundancy allocation $\left(l^*, r^*\right)$ of \eqref{eq:opt_prob} for finite-length codes. We emphasize that the optimal redundancy allocation $\left(l^*, r^*\right)$ is equivalent to finding the optimal (i.e., minimizing the probability of recovery error) point in the capacity region in Fig.~\ref{fig:BDEC_region}.
	
	In order to solve the optimization problem of \eqref{eq:opt_prob}, we need a closed-form expression for $P\left( \widehat{\mathbf{m}} \ne \mathbf{m} \right)$. Unfortunately, it is difficult to obtain the exact expression of $P\left( \widehat{\mathbf{m}} \ne \mathbf{m} \right)$. Thus, we will consider an estimate $(\widehat{l}, \widehat{r})$ which minimizes an upper bound on $P\left( \widehat{\mathbf{m}} \ne \mathbf{m} \right)$. 
	
	
	\begin{theorem} \label{thm:BDEC_UB} 
An upper bound on $P\left( \widehat{\mathbf{m}} \ne \mathbf{m} \right)$ is given by
		\begin{align}
		P\left( \widehat{\mathbf{m}} \ne \mathbf{m} \right) &\le \sum_{u=d_0}^{n}{ \beta^{u} \left( 1 - \beta \right)^{n-u}}\sum_{w=d_0}^{u}{B_{0,w} \binom{n-w}{u-w}}  \nonumber \\
		&+  \sum_{e = d_1}^{n}{\alpha^e \left(1 - \alpha \right)^{n-e}
			\sum_{w=d_1}^{e}{A_{w} \binom{n-w}{e-w}}}.
		\end{align} 
		where $A_w$ is the weight distribution of $\mathcal{C}$. If $A_w$ and $B_{0, w}$ follow the binomial distribution, 
		\begin{equation} \label{eq:BDEC_UB}
		P\left( \widehat{\mathbf{m}} \ne \mathbf{m} \right) \le 2^{-l} \left(1 + \beta \right)^n + 2^{-r} \left( 1 + \alpha \right)^n. 
		\end{equation}
	\end{theorem}
	\begin{proof}
		The proof is given in Appendix~\ref{pf:BDEC_UB}-F. 
	\end{proof}

	In this upper bound, we observe a dual relation between erasures and defects~\cite{Kim2016duality}. Note that $A_w$ and $B_{0,w}$ are the weight distributions of $\mathcal{C}$ and ${\mathcal{C}}_0^{\perp}$, respectively. 
		
	From these upper bounds on $P\left( \widehat{\mathbf{m}} \ne \mathbf{m} \right)$, we can modify \eqref{eq:opt_prob} for the BDEC as follows. 	

	\begin{equation}\label{eq:BDEC_opt_prob}
	\begin{aligned}
	( \widehat{l}, \widehat{r} )  = \; & \underset{(l, r)}{\text{argmin}} & & \text{Upper bound on }P\left( \widehat{\mathbf{m}} \ne \mathbf{m} \right)  \\
		& \text{subject to} & & 0 \le l \le n - k, \quad 0 \le r \le n - k\\
		&&&  l+r = n-k 
	\end{aligned}
	\end{equation}	

	
	If the code parameters $(n, k, l, d_0, d_1)$ and the channel parameters ($\alpha$, $\beta$) are given, the solution $( \widehat{l}, \widehat{r} )$ of \eqref{eq:BDEC_opt_prob} can be obtained. To illustrate this, we consider $\left[ n = 1023, k=923, l \right]$ PBCH codes. All possible redundancy allocation candidates of PBCH codes are listed in Table~\ref{tab:PLBC}. Since $l$ and $r$ are multiples of 10 (i.e., the degree of the Galois field of $[n = 1023, k]$ BCH codes), there are 11 redundancy allocation candidates. Hence, we can readily obtain the redundancy $( \widehat{l}, \widehat{r} )$ that minimizes the objective function of \eqref{eq:BDEC_opt_prob}.
		
	In addition, the objective function is \emph{convex} if we treat $l$ and $r$ as real values, even though we know that they are non-negative integers less than or equal to $n-k$. We can derive the solution $(\widetilde{l}, \widetilde{r})$ of \eqref{eq:BDEC_opt_prob} satisfying \emph{Karush-Kuhn-Tucker} (KKT) conditions.
	
	\begin{corollary} \label{cor:BDEC_KKT} Treating $l$ and $r$ as real values, the solution of \eqref{eq:BDEC_opt_prob} satisfying KKT conditions is given by
		\begin{numcases}{(\widetilde{l}, \widetilde{r})=}
		(0, n-k), & \hspace{-5mm} for $\frac{1+\alpha}{1+\beta} > 2^{1 - R} $ \label{eq:BDEC_opt_sol_con1}
		\\
		\left( \check{l}, \check{r}  \right),  & \hspace{-5mm} for $2^{-(1 - R)} \le \frac{1+\alpha}{1+\beta} \le 2^{1 - R} $ \label{eq:BDEC_opt_sol_con3}
		\\
		(n-k, 0), & \hspace{-5mm} for $\frac{1+\alpha}{1+\beta} < 2^{-(1 - R)}$ \label{eq:BDEC_opt_sol_con2}
		\end{numcases}
		where $\left( \check{l}, \check{r} \right)$ is given by
		\begin{align}
		\check{l} &= \frac{1}{2} \left\{ n \left( 1 - \log_2{\frac{1+\alpha}{1 +\beta}} \right) - k \right\}, \label{eq:BDEC_opt_sol_1} \\
		\check{r} &= \frac{1}{2} \left\{ n \left( 1 + \log_2{\frac{1+\alpha}{1 +\beta}} \right) - k \right\}. \label{eq:BDEC_opt_sol_2}
		\end{align}
	\end{corollary}
	\begin{proof}
		The proof is given in Appendix~\ref{pf:BDEC_KKT}-G.
	\end{proof}

	\begin{table}[t]
	\renewcommand{\arraystretch}{1.3}
	\caption{BDEC with the Same $C_{\text{BDEC}}=0.96$}
	\label{tab:BDEC_channel}
	\centering
	{\hfill{}
		\small{
			\begin{tabular}{c|c|c|c}
				\hline
				Channel & {$\alpha$} & {$\beta$} & {Remarks} \\ \hline \hline
				1       & 0.040 & 0 & BEC \\ \hline
				2       & 0.035 & 0.005 &  \\ \hline
				3       & 0.025 & 0.015 &  \\ \hline
				4       & 0.020 & 0.020 &  \\ \hline
				5       & 0.015 & 0.025 &  \\ \hline
				6       & 0.005 & 0.035 &  \\ \hline
				7       & 0 & 0.040 & BDC \\  \hline
		\end{tabular}}
	}
	\hfill{}
	\vspace{-3mm}
\end{table}
		
	If $\alpha$ is much larger than $\beta$ such that $\frac{1+\alpha}{1+\beta} > 2^{1 - R} $ for all possible $(l, r)$, then \eqref{eq:BDEC_opt_sol_con1} shows that we should allot all redundancy for correcting erasures to minimize the upper bound of \eqref{eq:BDEC_opt_prob}. If $\beta$ is much larger than $\alpha$ such that $\frac{1+\alpha}{1+\beta} < 2^{-(1 - R)}$ for all possible $(l, r)$, then \eqref{eq:BDEC_opt_sol_con2} shows that we should allot all redundancy for masking stuck-at defects to minimize the upper bound. For other $\alpha$ and $\beta$, we should allot the redundancy $(l, r)$ such that $2^{-l} \left(1 + \beta \right)^n = 2^{-r} \left( 1 + \alpha \right)^n$ to minimize the upper bound, which are satisfied by \eqref{eq:BDEC_opt_sol_1} and \eqref{eq:BDEC_opt_sol_2}.   
	
	\begin{remark}
		If only the normal cells can be erased, the corresponding channel's erasure probability is $\widetilde{\alpha} = (1-\beta)\alpha$. Hence the capacity will be $1 - \widetilde{\alpha} - \beta = (1 - \beta)(1-\alpha)$ which is the same as \eqref{eq:capacity_BDEC_tilde}. The equivalent results of Theorem~\ref{thm:BDEC_region}, Theorem~\ref{thm:BDEC_UB}, and Corollary~\ref{cor:BDEC_KKT} can be obtained by replacing $\alpha$ by $\widetilde{\alpha}$.
	\end{remark}

	In order to compare the optimal redundancy and the estimated redundancy allocation based on the derived upper bound, we consider several channels in Table~\ref{tab:BDEC_channel} whose capacities are all equal, i.e., $C_{\text{BDEC}}^{\text{enc}}=0.96$. For these channels, we investigate the performance of $\left[ n = 1023, k=923, l \right]$ PBCH codes in Table~\ref{tab:PLBC}.
	
	Fig.~\ref{fig:plot_BDEC_RA} shows the simulation results for the channels in Table~\ref{tab:BDEC_channel}. The simulation results of channel 1 (BEC) and channel 7 (BDC) are omitted because their optimal redundancy allocations are obvious. The optimal redundancy allocation for channel 1 (BEC) is $\left(l^*, r^* \right) = (0, 100)$. The more stuck-at defects a channel has, the larger $l$ is expected to be for the optimal redundancy allocation. Eventually, the optimal redundancy allocation for channel 7 (BDC) will be $\left(l^*, r^* \right) = (100, 0)$. The optimal $\left(l^*, r^*\right)$ can be obtained from Monte-Carlo simulation results in Fig.~\ref{fig:plot_BDEC_RA}, which are presented in the second column of Table~\ref{tab:BDSC_redundancy}.
	
	
	We can readily obtain the redundancy $(\widehat{l}, \widehat{r})$ that minimizes the upper bounds in Fig.~\ref{fig:plot_BDEC_RA}. The estimates of redundancy allocation $(\widehat{l}, \widehat{r})$ are shown in the third column of Table~\ref{tab:BDEC_redundancy} which shows that $(\widehat{l}, \widehat{r})$ by the upper bounds matches very well the redundancy $(l^*, r^*)$ determined by Monte-Carlo simulations.
	
	Next, $(\widetilde{l}, \widetilde{r})$ can be calculated from \eqref{eq:BDEC_opt_sol_con1}--\eqref{eq:BDEC_opt_sol_2} by treating $l$ and $r$ as real values. Resulting $(\widetilde{l}, \widetilde{r})$ values are shown in the last column of Table~\ref{tab:BDEC_redundancy}. $(\widehat{l}, \widehat{r} )$ is the nearest one from $(\widetilde{l}, \widetilde{r})$ considering the possible redundancy allocation candidates in Table~\ref{tab:PLBC}.	
	
	\begin{figure}[!t]
		\centering
		\includegraphics[width=0.4\textwidth]{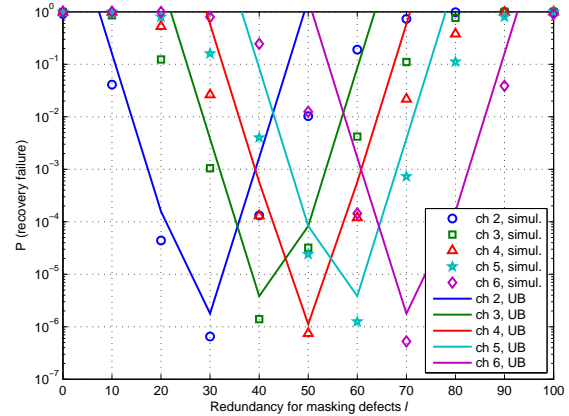}
		\caption{Comparison of simulation results (simul.) and upper bounds (UB) of the probability of recovery failure $P(\widehat{\mathbf{m}} \ne \mathbf{m})$ of BDEC channels in Table \ref{tab:BDEC_channel}.}
		\label{fig:plot_BDEC_RA}
	\end{figure}	

	\begin{table}[t]
		\renewcommand{\arraystretch}{1.3}
		\caption{Optimal Redundancy Allocations $\left(l^*, r^*\right)$ and Their Estimates $(\widehat{l}, \widehat{r})$ and $(\widetilde{l}, \widetilde{r})$ of BDEC}
		\label{tab:BDEC_redundancy}
		\centering
		\small{
		\begin{tabular}{c|c|c|c}
			\hline
			Channel & {$\left(l^*, r^*\right)$} & {$(\widehat{l}, \widehat{r})$} & {$(\widetilde{l}, \widetilde{r})$}    \\ \hline \hline
			1 & (0, 100) & (0, 100) & (0, 100)      \\ \hline
			2 & (30, 70) & (30, 70) & (28.3, 71.7)  \\ \hline
			3 & (40, 60) & (40, 60) & (42.8, 57.2)  \\ \hline
			4 & (50, 50) & (50, 50) & (50, 50)  \\ \hline
			5 & (60, 40) & (60, 40) & (57.2, 42.8)  \\ \hline
			6 & (70, 30) & (70, 30) & (71.7, 28.3)  \\ \hline
			7 & (100, 0) & (100, 0) & (100, 0)      \\ \hline
		\end{tabular}}
		\vspace{-5mm}
	\end{table}
	
\subsection{Redundancy Allocation: BDSC} \label{sec:BDSC_ra}

	A similar approach to redundancy allocation of the BDEC will be used for the BDSC. Instead of minimizing the upper bound on $P\left( \widehat{\mathbf{m}} \ne \mathbf{m} \right)$, we will derive an estimate of $P\left( \widehat{\mathbf{m}} \ne \mathbf{m} \right)$ for the BDSC and the redundancy allocation $(l, r)$ that minimizes this estimate of $P\left( \widehat{\mathbf{m}} \ne \mathbf{m} \right)$ will be used. 
	
	
	
	For the BDC and the BDEC, the number of unmasked defects after encoding failure (i.e., $\mathbf{c} \circ \mathbf{s} \ne \mathbf{c}$) does not affect $P\left( \widehat{\mathbf{m}} \ne \mathbf{m} \right)$. On the other hand, the number of unmasked defects is important for $P\left( \widehat{\mathbf{m}} \ne \mathbf{m} \right)$ of the BDSC. The reason is that the stuck-at errors due to unmasked defects can be treated as random errors and corrected during the decoding stage. Hence, we propose a two-step encoding method (described in Algorithm~\ref{alg:twostep}), which reduces the performance gap between \eqref{eq:BDC_opt_problem} and \eqref{eq:BDC_LE}. 
	
	\begin{algorithm}
		\caption{Two-step Encoding}\label{alg:twostep}
		\begin{algorithmic}[]
			\State \textbf{Step 1:} Try to solve \eqref{eq:BDC_LE}, i.e., $G_0^{\mathcal{U}} {\mathbf{d}} = {\mathbf{b}}^{\mathcal{U}}$.	
			\If{$u < d_0$} go to \textbf{End}. \Comment A solution $\mathbf{d}$ always exists.
			\Else
			\Comment A solution ${\mathbf{d}}$ exists so long as \eqref{eq:BDC_LE_sol_exist} holds.
			\If{$\mathbf{d}$ exists} go to \textbf{End}.
			\Else ~go to \textbf{Step 2}.
			\EndIf
			\EndIf
			\State \textbf{Step 2:}
			\begin{itemize}
				\item Choose $d_0 - 1$ locations among $\mathcal{U}$ and define ${\mathcal{U}}' = \left\{i_1, \ldots, i_{d_0 - 1}\right\}$.
				\item Solve the following linear equation: $G_0^{\mathcal{U}'}{\mathbf{d}} = {\mathbf{b}}^{\mathcal{U}'}$ \qquad \qquad \Comment A solution $\mathbf{d}$ always exists.
			\end{itemize}
			\State \textbf{End}					
		\end{algorithmic}
	\end{algorithm}		
	
	The two-step encoding tries to reduce the number of unmasked defects by using the second step (i.e., Step 2) even though $\mathbf{c} \circ \mathbf{s} \ne \mathbf{c}$. When the encoder fails to solve \eqref{eq:BDC_LE}, the encoder randomly chooses $d_0 - 1$ defect locations among $\mathcal{U}$ and define ${\mathcal{U}}' = \left\{i_1, \ldots, i_{d_0 - 1}\right\}$. Afterwards, the encoder solves $G_0^{\mathcal{U}'}{\mathbf{d}} = {\mathbf{b}}^{\mathcal{U}'}$ where a solution $\mathbf{d}$ always exists by the definition of $d_0$ in \eqref{eq:PLBC_d0}. If $\mathbf{d}$ is obtained in Step 2, then the number of unmasked defects is $u - \left(d_0 - 1 \right)$ instead of $u$. For the BDC and the BDEC, Step 2 cannot improve $P(\widehat{\mathbf{m}} \ne \mathbf{m})$ since unmasked stuck-at defects cannot be corrected during the decoding stage. However, Step 2 is helpful for the BDSC since the stuck-at errors can be regarded as random errors and corrected at the decoder.  
			
	The two-step encoding's complexity is ${\mathcal{O}}\left( u^3\right)$ because both Step 1 and Step 2 are related to solving the linear equations. Also, the bounded distance decoding for estimating $\widehat{\mathbf{z}}$ can be implemented by polynomial decoding algorithms. For PBCH codes, standard algorithms such as Berlekamp-Massey algorithm can be used for decoding. The flow of PLBC's decoding for the BDSC was explained in Section \ref{subsec:PLBC}. 

	We will derive the upper bound on $P\left( \widehat{\mathbf{m}} \ne \mathbf{m} \right)$ where the two-step encoding and the bounded distance decoding are used. 
	
	\begin{theorem}\label{thm:BDSC_bound_fail} The upper bound on $P\left( \widehat{\mathbf{m}} \ne \mathbf{m} \right)$ is given by
		\begin{align}
		\hspace{-3mm} & P\left( \widehat{\mathbf{m}}\ne \mathbf{m} \right) \nonumber  \\
		\hspace{-3mm} & \le \left[ \sum_{u=d_0}^{n} \left\{ \hspace{-1mm} { \binom{n}{u} \beta^{u} \left( 1 - \beta \right)^{n-u}} \min \left\{ \frac{ \sum_{w=d_0}^{u}{B_{0,w} \binom{n-w}{u-w}}}{\binom{n}{u}}, 1\right\}  \right. \right. \nonumber\\
		\hspace{-3mm} & \left. \left.\cdot \sum_{t=t_1 + d_0 - u}^{n}{\binom{n}{t} p^{t} \left( 1 - p \right)^{n-t}} \right\} \right] + \sum_{t=t_1 + 1}^{n}{\binom{n}{t}p^t (1-p)^{n-t}}
		\label{eq:BDSC_bound_fail}
		\end{align}
		where $t_1 = \left\lfloor \frac{d_1 - 1}{2} \right\rfloor$ is the error correcting capability of $\mathcal{C}$.
	\end{theorem}
	\begin{proof}The proof is given in Appendix~\ref{pf:BDSC_bound_fail}-H.
	\end{proof}
	
	

	During the derivation of the upper bound of \eqref{eq:BDSC_bound_fail}, we regard all the unmasked stuck-at defects as random errors. However, on average, only half of the unmasked defects result in error if $P(S=0) = P(S = 1) = \frac{\beta}{2}$. Thus, we can derive the following estimate of $P\left( \widehat{\mathbf{m}} \ne \mathbf{m} \right)$. 
	
	\begin{corollary}\label{cor:BDSC_estimate_fail} The estimate of $P\left( \widehat{\mathbf{m}} \ne \mathbf{m} \right)$ is given by
		\begin{align}
		& \hspace{-1mm} P\left( \widehat{\mathbf{m}} \ne \mathbf{m} \right) \nonumber \\
		& \hspace{-1mm} \simeq \left[ \sum_{u=d_0}^{n} \left\{ {\binom{n}{u} \beta^{u} \left( 1 - \beta \right)^{n-u}} \min\left\{ \frac{\sum_{w=d_0}^{u}{B_{0,w} \binom{n-w}{u-w}}}{\binom{n}{u}} , 1 \right\} \right.\right. \nonumber \\
		& \quad \cdot \left.\left. \sum_{t=t_1 - \left\lceil \frac{u - d_0 + 1}{2} \right\rceil + 1 }^{n}{\binom{n}{t} p^{t} \left( 1 - p \right)^{n-t}} \right\} \right]\nonumber \\
		& + \sum_{t=t_1 + 1}^{n}{\binom{n}{t}p^t (1-p)^{n-t}} \label{eq:BDSC_estimate_fail}
		\end{align}
		where $\lceil x \rceil$ is the smallest integer not less than $x$.
	\end{corollary}
	\begin{proof}
	The proof is given in Appendix~\ref{pf:BDSC_bound_fail}. 
	\end{proof}
	
	In order to compare the optimal redundancy $(l^*, r^*)$ and the estimated redundancy allocation, we consider the several channels shown in Table \ref{tab:BDSC_channel}. Channel 1 and Channel 7 are equivalent to the BSC and the BDC, respectively. For the other channels from Channel 2 to Channel 6, their lower bounds of the capacity $C_{\text{BDSC}}^{\text{lower}}$ of \eqref{eq:capacity_lower_BDSC} and the upper bounds $C_{\text{BDSC}}^{\text{upper}}$ are almost the same. Thus, we can estimate $C_{\text{BDSC}}^{\text{enc}}$ from bounds although the closed-form of $C_{\text{BDSC}}^{\text{enc}}$ is not known. 
	
	
	It is worth mentioning that all the channel parameters are chosen to have almost the same $\widetilde{p} \simeq 4 \times 10^{-3}$, which was given by~\eqref{eq:p_tilde}. Hence, all the channels show similar $P\left( \widehat{\mathbf{m}} \ne \mathbf{m} \right)$ for $(l, r) = (0, 100)$ which represents the case when the defect information is not used. On the other hand, each channel has different $C_{\text{BDSC}}^{\text{enc}}$. The larger $\beta$, the more defect information we can obtain, which results in the larger $C_{\text{BDSC}}^{\text{enc}}$. We apply $\left[ n = 1023, k=923, l \right]$ PBCH codes in Table \ref{tab:PLBC}.	
	
	Fig.~\ref{fig:plot_BDSC_RA} compares the simulation results and the estimates of $P\left( \widehat{\mathbf{m}} \ne \mathbf{m} \right)$ for the channels in Table~\ref{tab:BDSC_channel}, which shows that the estimates of $P\left( \widehat{\mathbf{m}} \ne \mathbf{m} \right)$ given by \eqref{eq:BDSC_estimate_fail} match well with simulation results. Hence, we can choose the redundancy allocation minimizing the estimates instead of the simulation results in spite of the binomial approximation of \eqref{eq:B0w_approximate}. The redundancy allocation that minimizes the estimate of $P\left( \widehat{\mathbf{m}} \ne \mathbf{m} \right)$ is the estimate of the optimal redundancy allocation, i.e., $(\widehat{l}, \widehat{r})$. The simulation results of $P\left( \widehat{\mathbf{m}} \ne \mathbf{m} \right) < 10^{-8}$ are incomplete because of their impractical computations. 
	
	Table~\ref{tab:BDSC_redundancy} shows that the estimate of the optimal redundancy allocation $(\widehat{l}, \widehat{r})$ is the same as the optimal redundancy allocation $(l^*, r^*)$ for the channels in Table~\ref{tab:BDSC_channel}. Thus, we can accurately estimate the optimal redundancy allocation without simulations. The estimate of optimal redundancy allocation requires much less computations than Monte-Carlo simulations.

	\begin{table}[t]
		\renewcommand{\arraystretch}{1.3}
		\caption{Several Channel Parameters of the BDSC}
		\label{tab:BDSC_channel}
		\centering
		{\small
			\begin{tabular}{c|c|c|c|c}
				\hline
				Channel & {$p$}                & {$\beta$}            & $C_{\text{BDSC}}^{\text{lower}}$ & $C_{\text{BDSC}}^{\text{upper}}$  \\ \hline \hline
				1       & $4.0 \times 10^{-3}$ & 0                    & \multicolumn{2}{c}{0.9624} \\ \hline
				2       & $3.0 \times 10^{-3}$ & $2.0 \times 10^{-3}$ & 0.9685 & 0.9686        \\ \hline
				3       & $2.5 \times 10^{-3}$ & $3.0 \times 10^{-3}$ & 0.9718 & 0.9719        \\ \hline
				4       & $2.0 \times 10^{-3}$ & $4.0 \times 10^{-3}$ & 0.9752 & 0.9753        \\ \hline
				5       & $1.0 \times 10^{-3}$ & $6.0 \times 10^{-3}$ & 0.9826 & 0.9827        \\ \hline
				6       & $5.0 \times 10^{-4}$ & $7.0 \times 10^{-3}$ & 0.9868 & 0.9868        \\ \hline
				7       &  0                   & $8.0 \times 10^{-3}$ & \multicolumn{2}{c}{0.9920}        \\ \hline
			\end{tabular}
		}
		\vspace{-2mm}
	\end{table}	
	
	\begin{figure}[!t]
		\centering
		\includegraphics[width=0.4\textwidth]{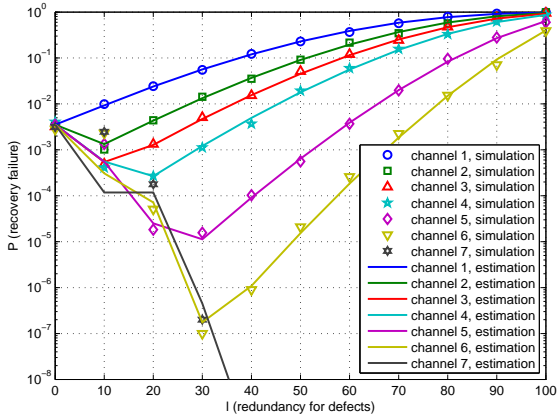}
		\caption{Comparison of simulation results and estimates of the probability of recovery failure $P\left( \widehat{\mathbf{m}} \ne \mathbf{m} \right)$.}
		\label{fig:plot_BDSC_RA}
		\vspace{-3mm}
	\end{figure}
	
	\begin{table}[t]
	\renewcommand{\arraystretch}{1.3}
	\caption{Optimal Redundancy Allocations $(l^*, r^*)$ and Their Estimates $(\widehat{l}, \widehat{r})$ of BDSC by~\eqref{eq:BDSC_estimate_fail}}
	\label{tab:BDSC_redundancy}
	\centering
	{\small
		\begin{tabular}{c|c|c}
			\hline
			Channel & $(l^*, r^*)$ & $(\widehat{l}, \widehat{r})$ \\ \hline \hline
			1 & (0, 100) & (0, 100)    \\ \hline
			2 & (10, 90) & (10, 90)   \\ \hline
			3 & (10, 90) & (10, 90)  \\ \hline
			4 & (20, 80) & (20, 80)   \\ \hline
			5 & (30, 70) & (30, 70)   \\ \hline
			6 & (30, 70) & (30, 70)   \\ \hline
			7 & (100, 0) & (100, 0) \\ \hline
		\end{tabular}}
		\vspace{-5mm}
	\end{table}
	
	Fig.~\ref{fig:plot_BDSC_RA} also shows that the optimal redundancy allocation significantly improves $P\left( \widehat{\mathbf{m}} \ne \mathbf{m} \right)$. For example, $P\left( \widehat{\mathbf{m}} \ne \mathbf{m} \right)$ of channel 6 is $1.00 \times 10^{-7}$ with the optimal redundancy allocation $(l^*, r^*) = (30, 70)$ whereas $P\left( \widehat{\mathbf{m}} \ne \mathbf{m} \right)$ is $2.80 \times 10^{-3}$ with $(l, r) = (0, 100)$. Thus, it is important to find and use the optimal redundancy allocation for better $P\left( \widehat{\mathbf{m}} \ne \mathbf{m} \right)$.
	
	In addition, it is worth mentioning that $P\left( \widehat{\mathbf{m}} \ne \mathbf{m} \right)$ for $(l^*, r^*)$ improves as $\beta$ increases in Fig.~\ref{fig:plot_BDSC_RA}. $P\left( \widehat{\mathbf{m}} \ne \mathbf{m} \right)$ of all the channels are almost the same for the redundancy allocation of $\left(l=0, r=n-k\right)$ because all the channel parameters are chosen to have the same $\widetilde{C}_{\text{BDSC}}^{\text{min}}$. As $\beta$ increases, $P\left( \widehat{\mathbf{m}} \ne \mathbf{m} \right)$ for $(l^*, r^*)$ improves because the PLBC can exploit more defect information as indicated by the lower and upper bounds on $C_{\text{BDSC}}^{\text{enc}}$ in Table \ref{tab:BDSC_channel}. 

\section{\uppercase{Conclusion}} \label{sec:conclusion}

	The redundancy allocation of finite-length codes for memory with permanent stuck-at defects and transient errors was investigated. We derived the upper bound on the probability of recovery failure for the BDEC and the estimate of the probability of recovery failure for the BDSC. Based on these analytical results, we estimated the optimal redundancy allocation effectively. The estimated redundancy allocation matches the optimal redundancy allocation well while requiring much less computation than Monte-Carlo simulations.

\appendices

\section*{A. Proof of Proposition~\ref{thm:capacity_BDEC}\label{pf:capacity_BDEC}}

	By Gelfand-Pinsker theorem~\cite{Gelfand1980}, 
	\begin{equation}
		C_{\text{BDEC}}^{\text{enc}} = \max_{P(U \mid S), X(U, S)}{(I(U; Y) - I(U; S))}
	\end{equation}
	where $\left| \mathcal{U} \right| \le \min{ \left\{ \left| \mathcal{X} \right| \cdot \left| \mathcal{S} \right|, \left| \mathcal{Y} \right| + \left| \mathcal{S} \right| - 1 \right\}}$. 
	It is clear that $I(U; Y) - I(U; S) = H(U \mid S) - H(U \mid Y)$. For $S = \lambda$, set $U \sim \text{Bern}(1/2)$ and $X=U$, i.e., $U$ is a Bernoulli random variable with parameter $\frac{1}{2}$. If $S \ne \lambda$, we set $U=X=S$. Then, $H(U \mid S) = H(X \mid S)  = 1 - \beta$. In addition, $	H(U\mid Y) = H(X) - I(X; Y) = \alpha$ where $H(X) = 1$ follows from $P(X = 0) = P(X = 1)$ for both stuck-at defects and normal cells. We should minimize $H(U \mid Y)$ to obtain $C_{\text{BDEC}}^{\text{enc}}$, which can be achieved by setting $I(X;Y) = 1 - \alpha$ (i.e., the capacity of the BEC). Hence, $C_{\text{BDEC}}^{\text{enc}} = 1 - \alpha - \beta$. 

\section*{B. Proof of Proposition~\ref{thm:capacity_bound_BDSC} \label{pf:capacity_bound_BDSC}}
	The upper bound of \eqref{eq:capacity_bound_BDSC} is easy to see because of \eqref{eq:capacity_BDSC_tilde} where we assume that the stuck-at defects do not suffer from random errors. Unlike the BDEC, the unmasked defects in the BDSC can be corrected by the decoder. Hence, we can allow the encoder to mask a fraction of stuck-at defects. Suppose that $\eta$ denotes the fraction of unmasked stuck-at defects during encoding. We need to find the optimal $\eta^*$, which makes it complicated to derive the capacity in~\cite{Heegard1983plbc}. By setting $\eta = 0$ instead of using the optimal $\eta^*$, the lower bound of \eqref{eq:capacity_lower_BDSC} can be derived, which is similar to the proof of Proposition \ref{thm:capacity_BDEC}. For $S = \lambda$, set $U \sim \text{Bern}(1/2)$ and $X=U$. For $S \ne \lambda$, we set $U=X=S$ which means that $\eta = 0$. Then,  $H(U \mid S) = 1 - \beta$. In addition, $H(U\mid Y) = H(X \mid Y) = h(p)$. Hence, $C_{\text{BDSC}}^{\text{lower}} = 1 - \beta - h(p)$. 

\section*{C. Proof of Proposition ~\ref{thm:BDC_capacity} \label{pf:BDC_capacity}}
		Suppose that each element of $G_0$ is selected at random with equal probability from $\{0, 1\}$. Then,
\begin{align}
P\left( E=0 \right) & = P\left( E=0, n(\beta - \epsilon) \le |{\mathcal{U}}| \le n(\beta + \epsilon) \right) + \epsilon' \nonumber\\
& \le \sum_{u=n(\beta - \epsilon)}^{n(\beta + \epsilon)}{P\left(\rank\left(G_0^{{\mathcal{U}}}\right) < u \mid |{\mathcal{U}}|=u\right)} + \epsilon'  \nonumber\\
& \le (2n\epsilon + 1) 2^{n \left( \frac{k}{n}  - (1 - \beta) + \epsilon\right)} + \epsilon' \label{eq:BDC_CA}
\end{align}
where we assume that $n(\beta \pm \epsilon)$ are integers without loss of generality. If $R = \frac{k}{n} < C_{\text{BDC}} - \epsilon$, $P\left( E=0 \right)$ converges to zero as $n \rightarrow \infty$. 

\section*{D. Proof of Lemma~\ref{lemma:BDC_ef_UB_fixed_u}\label{pf:BDC_ef_UB_fixed_u}}

	First, we will show that
	\begin{align}
		&P\left(E=0 \mid U=u\right)  \nonumber \\
		&= \sum_{j=1}^{u}{\left(1 - \frac{1}{2^j} \right)P \left( \rank \left( G_0^{\mathcal{U}} \right) =u - j \mid U=u \right)}.\label{eq:BDC_exact}
	\end{align}
	If $\rank \left( G_0^{\mathcal{U}} \right) = u$, \eqref{eq:BDC_LE} has at least one solution since \eqref{eq:BDC_LE_sol_exist} holds, i.e., $P\left(E=0 \mid U=u\right)=0$. If $\rank \left( G_0^{\mathcal{U}} \right) = u - j$ for $1 \le j \le u$, the last $j$ rows of the row reduced echelon form of $G_0^{\mathcal{U}}$ are zero vectors. In order to satisfy \eqref{eq:BDC_LE_sol_exist}, the last $j$ elements of the column vector $\mathbf{b}^{\mathcal{U}}$ should also be zeros. The probability that the last $j$ elements of $\mathbf{b}^{\mathcal{U}}$ are zeros is $\frac{1}{2^j}$ since $P(S=0 \mid S \ne \lambda) = P(S=1 \mid S \ne \lambda) = \frac{1}{2}$. Thus, $P\left(E=0 \mid U=u\right)$ is given by \eqref{eq:BDC_exact}. Also,
	\begin{align}
		\frac{P \left( \rank \left( G_0^{\mathcal{U}} \right) < u \mid U=u \right)}{2}
		 \le P\left(E=0 \mid U=u\right) \nonumber \\
		 \le P \left( \rank \left( G_0^{\mathcal{U}} \right) < u \mid U=u \right). \label{eq:enc_fail_upper0}
	\end{align}
	
	Suppose that there exists a nonzero codeword ${\mathbf{c}}^{\perp} \in {\mathcal{C}}_{0}^{\perp}$ of Hamming weight $w$. Note that $G_0$ is the parity check matrix of ${\mathcal{C}}_{0}^{\perp}$. Let $\Psi_w({\mathbf{c}}^{\perp})=\left\{ i \mid c_i^{\perp} \ne 0 \right\}$ denote the locations of nonzero elements of ${\mathbf{c}}^{\perp}$ and ${\mathcal{U}} = \left\{i_1,\ldots,i_u \right\}$ denote the locations of $u$ defects.
	
	If $\Psi_w({\mathbf{c}}^{\perp}) \subseteq \mathcal{U}$, $\rank \left( G_0^{\mathcal{U}} \right) < u$. Note that $G_0^{\Psi_{w}({\mathbf{c}}^{\perp})}$ is a submatrix of $G_0^{\mathcal{U}}$ and the rows of $G_0^{\Psi_{w}({\mathbf{c}}^{\perp})}$ are linearly dependent since $G_0^T {\mathbf{c}}^{\perp} = {\mathbf{0}}$.
	
	For any $\mathbf{c}^{\perp}$ such that $\Psi_w({\mathbf{c}}^{\perp}) \subseteq {\mathcal{U}}$, the number of possible $\mathcal{U}$ is $\binom{n-w}{u-w}$. Due to double counting, the number of $\mathcal{U}$ which results in $\rank \left( G_0^{\mathcal{U}} \right) < u$  will be less than or equal to $\sum_{w=d_{0}}^{u}{B_{0,w} \binom{n-w}{u-w}}$. Since the number of all possible $\mathcal{U}$ such that $U=u$ is $\binom{n}{u}$,
	\begin{equation} \label{eq:BDC_upper_bound_rank}
		P\left(\rank \left( G_0^{\mathcal{U}} \right) < u \mid U=u \right) \le \frac{\sum_{w=d_{0}}^{u}{B_{0, w} \binom{n-w}{u-w}}}{\binom{n}{u}}.
	\end{equation}
	
	From \eqref{eq:enc_fail_upper0} and \eqref{eq:BDC_upper_bound_rank}, the upper bound on $P\left(E=0 \mid U=u \right)$ is given by \eqref{eq:enc_fail_upper}. Note that we set $P\left(E=0 \mid U=u \right) = 1$ if $\frac{\sum_{w=d_{0}}^{u}{B_{0, w} \binom{n-w}{u-w}}}{\binom{n}{u}} \ge 1$.

\section*{E. Proof of Lemma~\ref{lemma:BDC_ef_exact}\label{pf:BDC_ef_exact}}

	The proof has two parts. First, we will show that
	\begin{equation} \label{eq:BDC_exact_condition_pf1}
	P\left(\rank \left( G_0^{\mathcal{U}} \right) < u \mid U=u \right) = \frac{\sum_{w=d_{0}}^{u}{B_{0, w} \binom{n-w}{u-w}}}{\binom{n}{u}}
	\end{equation}
	for $u \le d_0 + t_0$ where $t_0 = \left\lfloor \frac{d_0 - 1}{2} \right\rfloor$, which means that there is no double counting in \eqref{eq:BDC_upper_bound_rank}. Second, we will prove that
	\begin{align}
	& P\left(\rank \left( G_0^{\mathcal{U}} \right) < u \mid U=u \right) \nonumber \\
	& \quad = P\left(\rank \left( G_0^{\mathcal{U}} \right) = u-1 \mid U=u \right) \label{eq:BDC_exact_condition_pf2}
	\end{align}
	for $u \le d_0 + t_0$, i.e., $P\left(\rank \left( G_0^{\mathcal{U}} \right) \le u-2 \mid U=u \right)=0$.
	
	Then, $P\left(E=0 \mid U=u \right)$ is given by
	\begin{align}
	\hspace{-5mm}P\left(E=0 \mid U=u \right) 
		&= \frac{P \left( \rank \left( G_0^{\mathcal{U}} \right) = u - 1 \mid U=u \right)}{2}  \label{eq:BDC_exact_condition_step1}\\
	    &= \frac{1}{2} \cdot \frac{\sum_{w=d_{0}}^{u}{B_{0, w} \binom{n-w}{u-w}}}{\binom{n}{u}} \label{eq:BDC_exact_condition_step2}
	\end{align}
	where \eqref{eq:BDC_exact_condition_step1} follows from \eqref{eq:BDC_exact} and \eqref{eq:BDC_exact_condition_pf2}. Also, \eqref{eq:BDC_exact_condition_step2} follows from \eqref{eq:BDC_exact_condition_pf1}.
	
	1) Proof of \eqref{eq:BDC_exact_condition_pf1}
	
	Suppose that there are two nonzero codewords ${\mathbf{c}}_1^{\perp}, {\mathbf{c}}_2^{\perp} \in {\mathcal{C}}_0^{\perp}$ such that $\|{\mathbf{c}}_1^{\perp}\| = w_1$ and $\|{\mathbf{c}}_2^{\perp}\| = w_2$. Without loss of generality, assume that $d_0 \le w_1 \le w_2$. The locations of nonzero elements in ${\mathbf{c}}_1^{\perp}$ and ${\mathbf{c}}_2^{\perp}$ are given by $	\Psi_{w_1}\left({\mathbf{c}}_1^{\perp}\right)= \left\{ i_{1,1},\ldots,i_{1, w_1} \right\}$ and $	\Psi_{w_2}\left({\mathbf{c}}_2^{\perp}\right)= \left\{ i_{2,1},\ldots,i_{2, w_2} \right\}$. 
	
	Let $\Psi_{\alpha}=\left\{i_1, \ldots, i_{\alpha} \right\}$ denote $\Psi_{\alpha} = \Psi_{w_1}\left({\mathbf{c}}_1^{\perp}\right) \cap \Psi_{w_2}\left({\mathbf{c}}_2^{\perp}\right)$. Then,  $\Psi_{w_1}\left({\mathbf{c}}_1^{\perp}\right) = \Psi_{\alpha} \cup \left\{ i_{1,1}',\ldots,i_{1, \beta_1}' \right\}$ and $\Psi_{w_2}\left({\mathbf{c}}_2^{\perp}\right) = \Psi_{\alpha} \cup \left\{ i_{2,1}',\ldots,i_{2, \beta_2}' \right\}$ where $i_{1, j_1}'$ for $j_1 \in \left\{1, \ldots, \beta_1 \right\}$ and $i_{2, j_2}'$ for $j_2 \in \left\{1, \ldots, \beta_2 \right\}$ are the reindexed locations of nonzero elements of ${\mathbf{c}}_1^{\perp}$ and ${\mathbf{c}}_2^{\perp}$ that are mutually disjoint with $\Psi_{\alpha}$. Note that $\left\{ i_{1,1}',\ldots,i_{1, \beta_1}' \right\} \cap \left\{ i_{2,1}',\ldots,i_{2, \beta_2}' \right\} = \emptyset $, $\beta_1 = w_1 - \alpha$ and $\beta_2 = w_2 - \alpha$.
	
	Due to the property of linear codes, ${\mathbf{c}}_3^{\perp} = {\mathbf{c}}_1^{\perp} + {\mathbf{c}}_2^{\perp}$ is also a codeword of ${\mathcal{C}}_0^{\perp}$, i.e., ${\mathbf{c}}_3^{\perp} \in {\mathcal{C}}_0^{\perp}$ and $\|{\mathbf{c}}_3^{\perp} \|=\beta_1 + \beta_2$. Also, the following conditions hold because of the definition of $d_0$.
	\begin{equation*}
		\alpha + \beta_1 \ge d_0, \quad \alpha + \beta_2 \ge d_0, \quad 	\beta_1 + \beta_2 \ge d_0
	\end{equation*}
	Thus, we can claim that $2 \left( \alpha + \beta_1 + \beta_2 \right) \ge 3 d_0$, which results in $\alpha + \beta_1 + \beta_2 \ge d_0 + \left\lfloor \frac{d_0 + 1}{2} \right\rfloor = d_0 + t_0 + 1$ since $\alpha + \beta_1 + \beta_2$ has to be an integer.
	
	For double counting to occur in \eqref{eq:BDC_upper_bound_rank}, there should exist at least two codewords ${\mathbf{c}}_1^{\perp}$ and ${\mathbf{c}}_2^{\perp}$ such that $\Psi_{w_1}\left({\mathbf{c}}_1^{\perp}\right) \cup \Psi_{w_2}\left({\mathbf{c}}_2^{\perp}\right) \subseteq\mathcal{U}$. It means that double counting occurs only if $u \ge \alpha + \beta_1 + \beta_2 \ge d_0 + t_0 + 1$. Thus, there is no double counting for $u \le d_0 + t_0$. For $u \le d_0 + t_0$, there exists at most one codeword ${\mathbf{c}}^{\perp}$ such that $\Psi_{w}\left({\mathbf{c}}^{\perp} \right) \subseteq {\mathcal{U}}$.
	
	2) Proof of \eqref{eq:BDC_exact_condition_pf2}
	
	It is clear that $\rank \left( G_0^{\mathcal{U}}\right) = u - 1$ if and only if there exists only one nonzero codeword $\mathbf{c}^{\perp}$ such that $\Psi_{w}\left( {\mathbf{c}}^{\perp} \right) \subseteq {\mathcal{U}}$. Note that $\rank \left( G_0^{\mathcal{U}}\right) < u - 1$ if and only if $\mathcal{U}$ includes the locations of nonzero elements of at least two nonzero codewords. We have already shown that there exists at most one nonzero codeword ${\mathbf{c}}^{\perp}$ such that $\Psi_{w}\left( {\mathbf{c}}^{\perp} \right) \subseteq {\mathcal{U}}$ for $u \le d_0 + t_0$.

\section*{F. Proof of Theorem~\ref{thm:BDEC_UB}\label{pf:BDEC_UB}}	

	From $P\left( \widehat{\mathbf{m}}  \ne {\mathbf{m}} \right) = P\left(E=0\right) + P\left(E=1, D=0\right)$, we derive the upper bounds on $P\left(E=0\right)$ and $P\left(E=1, D=0\right)$ respectively. The upper bounds on $P\left(E=0\right)$ was shown in Corollary~\ref{cor:BDC_UB_1} and \ref{cor:BDC_UB_2}. If the additive encoding succeeds (i.e., $E=1$), then the corresponding channel is equivalent to the BEC with the erasure probability $\alpha$. The upper bound is given by
	\begin{align}
		P\left(E=1, D=0\right) & \le \sum_{e = d_1}^{n}{\alpha^e \left(1 - \alpha \right)^{n-e}
			\sum_{w=d_1}^{e}{A_{w} \binom{n-w}{e-w}}} \nonumber \\
		& = 2^{-r} \left(1 + \alpha \right)^n \label{eq:RA_D}
	\end{align}
which follows from $
		P\left(D=0 \mid |{\mathcal{E}}|=e \right) \le \frac{\sum_{w=d_{1}}^{e}{A_{w} \binom{n-w}{e-w}}}{\binom{n}{e}}$. If $A_w = 2^{-r} \binom{n}{w}$, \eqref{eq:RA_D} can be derived by a similar way as for Corollary~\ref{cor:BDC_UB_2}. 

\section*{G. Proof of Corollary~\ref{cor:BDEC_KKT}\label{pf:BDEC_KKT}}
	Suppose that $l$ and $r$ are real values. Since the objective function is convex and other constraints are linear, the optimization problem of \eqref{eq:BDEC_opt_prob} is convex. The Lagrangian $L$ is given by
	\begin{align}
	&L\left( l, r, \lambda_1, \lambda_2, \lambda_3, \lambda_4, \nu \right) = 2^{-l} \left(1 + \beta \right)^n + 2^{-r} \left( 1 + \alpha \right)^n \nonumber \\
	& + \lambda_1 (-l) + \lambda_2 (-r) + \lambda_3 \left\{l - (n - k)\right\} + \lambda_4 \left\{r - (n - k)\right\}  \nonumber \\
	& + \nu \left\{ l + r - (n - k) \right\}
	\end{align}
	where $\lambda_i$ for $i = 1, 2, 3, 4$ are the Lagrange multipliers associated with the inequality constraints and $\nu$ is the Lagrange multiplier with the equality constraint. 
	
	The KKT conditions can be derived as follows.	
	\begin{align}
	\nabla L & = 0 \label{eq:appendix_KKT_gradient_expand}\\
	-l & \le 0, \quad -r \le 0 \\
	l - (n - k) & \le 0,\quad r - (n - k) \le 0 \\
	l + r - (n - k ) & = 0 \label{eq:appendix_KKT_equality} \\
	\lambda_i & \ge 0, \quad i = 1,\ldots, 4 \\
	\lambda_1 l & = 0, \quad \lambda_2 r = 0 \\
	\lambda_3 \left\{ l - (n - k) \right\} & = 0, \quad \lambda_4 \left\{ r - (n - k) \right\} = 0
	\end{align}
		
	
	
	If $\alpha$ is much greater than $\beta$ such that $\frac{1+\alpha}{1+\beta} > 2^{1 - R} $, then we can claim that $2^{-l} \left(1 + \beta \right)^n < 2^{-r} \left(1 + \alpha \right)^n$ for any $(l, r)$ where $l+r = n-k$. By \eqref{eq:appendix_KKT_gradient_expand}, it is true that  
	\begin{equation}
		2^{-r} \left(1 + \alpha \right)^n - 2^{-l} \left(1 + \beta \right)^n = \lambda_1' - \lambda_2' - \lambda_3' + \lambda_4' > 0
	\end{equation}
	where $\lambda_i' = \frac{\lambda_i}{\ln{2}}$. Thus, $\lambda_1 + \lambda_4 > \lambda_2 + \lambda_3 \ge 0$. It is clear that $(l, r) = (0, n-k)$ satisfies the KKT conditions since $\lambda_2 = \lambda_3 = 0$ due to complement slackness. If $\beta$ is much greater than $\alpha$ such that $\frac{1+\alpha}{1+\beta} < 2^{-(1 - R)}$, then it can be shown that $(l, r) = (n-k, 0)$ satisfies the KKT conditions similarly. 
	
	Otherwise, suppose that $0<l<n-k$ and $0 < r < n-k$. Due to complementary slackness, $\lambda_1 = \lambda_2 = \lambda_3 = \lambda_4 = 0$. Thus, \eqref{eq:appendix_KKT_gradient_expand} will be as follows.
	\begin{align}
		- \ln{2} \cdot 2^{-l} \left( 1 + \beta \right)^n + \nu &= 0 \label{eq:appendix_opt_sol_con3_1}\\
		- \ln{2} \cdot 2^{-r} \left( 1 + \alpha \right)^n + \nu &= 0 \label{eq:appendix_opt_sol_con3_2}
	\end{align}
	By \eqref{eq:appendix_opt_sol_con3_1} and \eqref{eq:appendix_opt_sol_con3_2},
	\begin{equation} \label{eq:appendix_opt_sol_con3_3}
		2^{-l} \left( 1 + \beta \right)^n = 2^{-r} \left( 1 + \alpha \right)^n.
	\end{equation}
	Then, we can derive \eqref{eq:BDEC_opt_sol_1} and \eqref{eq:BDEC_opt_sol_2} by \eqref{eq:appendix_KKT_equality} and \eqref{eq:appendix_opt_sol_con3_3}. 		
\section*{H. Proof of Theorem~\ref{thm:BDSC_bound_fail} and Corollary~\ref{cor:BDSC_estimate_fail}\label{pf:BDSC_bound_fail}}

	The probability of recovery failure $P\left( \widehat{\mathbf{m}} \ne \mathbf{m} \right)$ is given by
	\begin{equation*}
		P\left( \widehat{\mathbf{m}} \ne {\mathbf{m}} \right) = P(E=0, \widehat{\mathbf{m}} \ne {\mathbf{m}}) + P(E=1, \widehat{\mathbf{m}} \ne \mathbf{m}).
	\end{equation*}
	The stuck-at errors due to unmasked defects can be corrected at the decoder even if the encoding fails. First, we will derive the upper bound on $P(E=0, \widehat{\mathbf{m}} \ne \mathbf{m})$. By the chain rule, $P(E=0, \widehat{\mathbf{m}} \ne \mathbf{m})$ is given by
	\begin{align}
	{P(E=0, \widehat{\mathbf{m}} \ne \mathbf{m})} &= \sum_{u=1}^{n} { \left\{ P(U=u) P(E=0 \mid U=u)  \right. } \nonumber \\
	& \quad \left. \cdot	{P(\widehat{\mathbf{m}} \ne {\mathbf{m}} \mid E=0,U=u)} \right\} \nonumber
	\end{align}
	where $P(U=u) = \binom{n}{u} \beta^{u} \left( 1 - \beta \right)^{n-u}$. In addition, the upper bound on $P(E=0 |U=u)$ is given by Lemma~\ref{lemma:BDC_ef_UB_fixed_u}. Also, $P(\widehat{\mathbf{m}} \ne {\mathbf{m}} \mid E=0, U=u)$ is given by
	\begin{align}
		{P(\widehat{\mathbf{m}} \ne {\mathbf{m}} \mid E=0, U=u)} &\le P\left( t \ge t_1 - u + d_0 \right). \label{eq:P_D0E0Uu1}
	\end{align}
	where $t$ is the number of random errors. Note that $u - \left( d_0 - 1 \right)$ represents the number of unmasked defects when two-step encoding fails. Since the number of random errors can be modeled by the binomial random variable, $P(\widehat{\mathbf{m}} \ne {\mathbf{m}} \mid E=0, U=u)$ is given by
	\begin{align}
	\hspace{-2mm}P\left( \widehat{\mathbf{m}} \ne {\mathbf{m}} \mid E=0, U=u  \right) 
	&\le \hspace{-5mm} \sum_{t=t_1 + d_0 - u}^{n}\hspace{-2mm}{\binom{n}{t} p^{t} \left( 1 - p \right)^{n-t}}. 
	\end{align}
	Hence, 
	\begin{align}
		&P(E=0, \widehat{\mathbf{m}} \ne {\mathbf{m}}) \nonumber \\
		& \le \sum_{u=d_0}^{n} \left\{ \binom{n}{u} {\beta^{u} \left( 1 - \beta \right)^{n-u}} {\min \left\{ \frac{\sum_{w=d_0}^{u}{B_{0,w} \binom{n-w}{u-w}}}{\binom{n}{u}} , 1\right\}} \right. \nonumber \\
		&\quad\left. \cdot \sum_{t=t_1 + d_0 - u}^{n}{\binom{n}{t} p^{t} \left( 1 - p \right)^{n-t}} \right\}.
	\end{align}
	
	Now, we will derive the upper bound on $P(E=1, \widehat{\mathbf{m}} \ne \mathbf{m})$. If the encoding succeeds, the channel is equivalent to the BSC.
	\begin{align}
		P(E=1, \widehat{\mathbf{m}} \ne {\mathbf{m}}) \le \sum_{t=t_1 + 1}^{n}{\binom{n}{t}p^t (1-p)^{n-t}}. 
	\end{align}
	
	For the proof of Corollary~\ref{cor:BDSC_estimate_fail}, we will change \eqref{eq:P_D0E0Uu1} by considering $P(S=0) = P(S=1) = \frac{\beta}{2}$. Because only the half of unmasked defects result in stuck-at errors on average, 
	\begin{align}
	& {P(\widehat{\mathbf{m}} \ne {\mathbf{m}} \mid E=0, U=u)} \nonumber \\
	& \simeq P\left( t \ge t_1 - \left\lceil \frac{u -  d_0 + 1}{2} \right\rceil + 1 \right)  
	\end{align}
	Thus, \eqref{eq:BDSC_bound_fail} will be changed into~\eqref{eq:BDSC_estimate_fail}.
	

%
%


\vspace{10pt}
\bibliographystyle{IEEE}

\end{document}